\documentclass[a4paper,noarxiv,onecolumn,accepted=2026-03-13]{quantumarticle}
\usepackage[utf8]{inputenc}
\usepackage[english]{babel}
\usepackage[T1]{fontenc}

\usepackage{lipsum}
\usepackage{mdframed}
\usepackage{booktabs}
\usepackage{longtable}

\usepackage{algpseudocode}
\usepackage{amsmath}
\usepackage{amssymb}
\usepackage{amsthm, thm-restate}
\usepackage{authblk}
\usepackage{bbm}
\usepackage{bm}
\usepackage{dsfont}
\usepackage{enumerate}
\usepackage{etoolbox}
\usepackage{geometry}
\usepackage{graphicx}
\usepackage{mathtools}
\usepackage{physics}
\usepackage{todonotes}
\usepackage{upgreek}
\usepackage[normalem]{ulem}
\usepackage{xcolor}

\usepackage{doi} 
\usepackage{nameref} 
\usepackage[numbers,sort&compress]{natbib}
\usepackage{hyperref}
\usepackage[noabbrev,capitalise]{cleveref}
    \crefname{figure}{Figure}{figures}

\newtheorem{theorem}{Theorem}
\newtheorem{proposition}[theorem]{Proposition}
\newtheorem{lemma}[theorem]{Lemma}
\newtheorem{corollary}[theorem]{Corollary}

\theoremstyle{definition}
\newtheorem{definition}[theorem]{Definition}
\theoremstyle{remark}
\newtheorem{remark}[theorem]{Remark}
\newtheorem{example}[theorem]{Example}

\def\Cl{\mathbb{C}}

\def\fin{\mathrm{fin}}

\def\P{\mathbb{P}}

\def\spann{\mathrm{span}}

\def\P{\mathcal{P}}
\def\M{\mathcal{M}}

\def\G{\mathfrak{G}}

\def\conv{\mathrm{Conv}}

\def\mS{\mathcal{S}}
\def\M{\mathcal{M}}
\def\mD{\mathcal{D}}
\def\E{\mathcal{E}}

\def\A{\mathsf{A}}
\def\B{\mathsf{B}}
\def\C{\mathsf{C}}
\def\D{\mathsf{D}}

\def\Q{\mathsf{Q}}

\def\R{\mathbb{R}}
\def\I{\mathbb{I}}
\def\N{\mathbb{N}}
\def\Nl{\mathbb{N}}
\def\Id{\mathbb{I}}

\def\mI{\mathcal{I}}
\def\mJ{\mathcal{J}}
\def\mK{\mathcal{K}}

\newcommand{\id}{\mathrm{id}}
\newcommand{\contord}{\succeq}
\newcommand{\newterm}[1]{\textbf{#1}}
\newcommand{\ph}{\mathord{\rule[-0.05em]{0.6em}{0.05em}}}		
\newcommand{\excess}{\varepsilon}
\newcommand{\qmarks}[1]{``#1''}
\newcommand{\tildetensor}{\mathbin{\tilde{\otimes}}}

\newcommand{\pomyield}{\mathfrak{p}_n{\textrm{-}}\mathrm{yield}}

	\let\abs\relax
	\let\norm\relax
	\DeclarePairedDelimiter{\abs}{\lvert}{\rvert}
	\DeclarePairedDelimiter{\norm}{\lVert}{\rVert}
	\DeclarePairedDelimiterXPP{\pnorm}[2]{}{\lVert}{\rVert}{_{#1}}{#2}
		
	\makeatletter
		\let\oldabs\abs
		\def\abs{\@ifstar{\oldabs}{\oldabs*}}
		\let\oldnorm\norm
		\def\norm{\@ifstar{\oldnorm}{\oldnorm*}}
		\let\oldpnorm\pnorm
		\def\pnorm{\@ifstar{\oldpnorm}{\oldpnorm*}}
	\makeatother
 
\providecommand{\given}{}		
\newcommand{\SetSymbol}[1][]{%
	\nonscript\;\,#1\vert
	\allowbreak
	\nonscript\;\,
	\mathopen{}
}
\let\Set\relax
\DeclarePairedDelimiterX{\Set}[1]{\{}{\}}{%
	\renewcommand{\given}{\SetSymbol[\delimsize]}
	#1
}
\makeatletter
\let\oldSet\Set
\def\Set{\@ifstar{\oldSet}{\oldSet*}}
\makeatother

\let\Family\relax
\DeclarePairedDelimiterX{\Family}[1]{(}{)}{%
	\renewcommand{\given}{\SetSymbol[\delimsize]} 
	#1
}
\makeatletter
\let\oldFamily\Family
\def\Family{\@ifstar{\oldFamily}{\oldFamily*}}

\makeatletter
    \patchcmd\NAT@citexnum{\let\NAT@last@num\NAT@num}{\MakeLinkTarget[cite]{}\Hy@backout{\@citeb\@extra@b@citeb}\let\NAT@last@num\NAT@num}{}{\fail}
\makeatother

\title{Resource-Theoretic Hierarchy of Contextuality for General Probabilistic Theories}

\author[1]{Lorenzo Catani} \email{lorenzo.catani4@gmail.com}
\author[2,3]{Thomas D. Galley} \email{thomas.galley@oeaw.ac.at}
\author[4]{Tom{\'a}{\v{s}} Gonda} \email{tomas.gonda@uibk.ac.at}

\affil[1]{International Iberian Nanotechnology Laboratory, Av. Mestre José Veiga s/n, 4715-330 Braga, Portugal}
\affil[2]{Institute for Quantum Optics and Quantum Information,
Austrian Academy of Sciences, Boltzmanngasse 3, A-1090 Vienna, Austria}
\affil[3]{Vienna Center for Quantum Science and Technology (VCQ), Faculty of Physics, University of Vienna, Vienna, Austria}
\affil[4]{Institute for Theoretical Physics, University of Innsbruck, Austria}

\begin{document}
\date{}

\maketitle

\begin{abstract} 
    
    In this work we present a hierarchy of generalized contextuality. 
    It refines the traditional binary distinction between contextual and noncontextual theories, and facilitates their comparison based on how contextual they are. 
    Our approach focuses on the contextuality of prepare-and-measure scenarios, 
    described by general probabilistic theories (GPTs).
    To motivate the hierarchy, we define it as the resource ordering of a novel resource theory of GPT-contextuality.
    The building blocks of its free operations are classical systems and univalent simulations between GPTs.
    These simulations preserve operational equivalences and thus cannot generate contextuality.
    Noncontextual theories can be recovered as least elements in the hierarchy. 
    We then define a new contextuality monotone, called classical excess, given by the minimal error of embedding a GPT within an infinite classical system. 
    In addition, we show that the optimal success probability in the parity oblivious multiplexing game also defines a monotone in our resource theory.
    Finally, we discuss whether the non-free operations can be understood as implementing information erasure and thus explaining the fine-tuning aspect of contextuality.
\end{abstract}

\newpage

\section{Introduction}
    
    \paragraph{Generalized (non)contextuality.} 
    A crucial research question in the foundations of quantum theory is to identify those features of quantum theory that constitute a true departure from the classical worldview. 
    Addressing this question requires one to first establish a good notion of classicality, which adequately captures the classical worldview. 
    We believe that a good notion of classicality should satisfy the following desiderata (see also \cite{SchmidSolstice2022}): (1) it endorses a principle that defines a clear boundary between aspects that pose interpretational issues and those that do not, (2) it is of broad range of applicability, (3) it is empirically testable, and (4) its violation constitutes a resource for practical applications, in particular in quantum information processing.\footnote{This last desideratum is motivated by the belief that identifying the true nonclassicality of quantum theory will ultimately provide the answer to the question about the origin of the alleged quantum computational speed-up.} 
    Motivated by these desiderata, a leading notion of classicality is \emph{generalized noncontextuality} \cite{Spekkens2005}.\footnote{In what follows we will often omit ``generalized'' and just talk of ``(non)contextuality''.} 
    A noncontextual theory is one that is compatible with a classical realist explanation of its operational predictions{\,---\,} it admits of a \emph{noncontextual ontological model}.
    In such a model, any two experimental procedures that the theory predicts to be operationally indistinguishable also have the same ontological representation (see \Cref{sec:prep_cont} for a precise definition).
    
    Noncontextuality satisfies (1), in that it is an instance of a methodological principle inspired by Leibniz's principle of the identity of indiscernibles \cite{Spekkens2019} (also formulated as a no fine-tuning principle \cite{CataniLeifer2020}). 
    A violation of such principle would indeed entail an interpretational problem, as it would attribute a conspiratorial connotation to the realist explanation of the theory: why should experimental procedures predicted by the theory to be indistinguishable \textit{in principle} be represented by different distributions in the ontological model?
    Noncontextuality also satisfies (2), as it applies to a wide range of scenarios including prepare-and-measure experiments of a single system, unlike Bell's local causality{\,---\,}another leading notion of classicality. 
    Moreover, in situations where these are applicable, it coincides with notions of classicality such as non-negativity of quasi-probability representations \cite{Ferrie2008,Spekkens2008} and Bell's local causality \cite{Bell1964}.
    It satisfies (3), as witnessed by the experiments performed to test quantum violations of generalized noncontextuality \cite{Mazurek2016,Mazurek2021}.
    Finally, desideratum (4) has been argued by the numerous works showing that contextuality is a resource for information processing tasks \cite{Spekkens2009,Haamedi2017,Schmid2018,SahaAnubhav2019,LostaglioSenno2020,lostaglio2020certifying,Yadavalli2020,Flatt2021,Roch2021,Catani2022UR,Wagner2022,Catani2023WP}.

    \paragraph{Hierarchy of contextuality.} 
    In relation to the last point, contextuality can be used to witness and characterize the advantage of using \emph{quantum} physical systems in practical applications.
    In such contexts, it is important to know not just whether a theory is contextual, but also to \textit{quantify} how contextual it is.
    Our article aims to address this by introducing a \emph{hierarchy of contextuality}, in which not all contextual phenomena are equivalent.
    In particular, it allows us to make more fine-grained distinctions and to quantify the amount of contextuality present in a theory.
    
    To motivate the hierarchy we propose, we take inspiration from the framework of resource theories \cite{Coecke2016,chitambar2019quantum,gour2024resources}. 
    There, a resource object $r$ (such as a quantum state) is deemed more valuable than a resource object $s$, denoted $s \preceq r$, if there is a \emph{free operation} that transforms $r$ to $s$.
    The preorder relation\footnotemark{} $\preceq$ is called the \emph{resource ordering}.
    \footnotetext{A preorder is a relation that is both reflexive and transitive. 
    It is a partial order if it is also anti-symmetric.
    If, in addition, any two elements are related, then it is a total order.
    }%
    A paradigmatic example is that of entanglement ordering of bipartite quantum states, where the free channels consist of local operations and classical communication~\cite{Horodecki2009}.
    In this sense, a resource theory (and its associated resource ordering) is defined by a choice of resource objects and of free operations~\cite{gonda2021resource}.
    For instance, if the resource objects studied are quantum states, then the set of free operations is typically a subset of quantum channels closed under composition.
    
    In our case, resource objects are instead physical theories, since (non)contextuality is a property of theories rather than states.
    However, as contextuality can be witnessed by individual systems within the theory, we restrict our investigation to individual systems.
    That is, we identify each resource object as a system in a general probabilistic theory (GPT)~\cite{Hardy2001,Barrett2007,Janotta2014,plavala2021general}, such as classical probability theory, quantum theory, or a subtheory thereof like the stabilizer subtheory \cite{Gottesman1997}.
    The GPT system specifies all possible probabilistic behaviours of this physical system and thus characterizes its information-theoretic properties.
    
    Since our resources are not states of a physical system, but rather physical theories, this implies that the transformations that we consider cannot be standard physical operations, as is the case for traditional resource theories of quantum states. 
    Our operations of interest are \emph{simulations}, which are transformations between theories that faithfully encode the information of one theory within another.
    Simulations that preserve indistinguishability have been introduced in \cite{MullerGarner2021} as \emph{univalent} simulations.
    Since noncontextuality of an ontological model means that the model corresponds to a univalent simulation, one can show \cite{MullerGarner2021} that a GPT system is noncontextual if and only if it can be simulated in a univalent way by a classical GPT system{\,---\,}one, whose states form a simplex of probability distributions on a finite sample space.
    Classical GPT systems thus cannot generate contextuality under univalent simulations.
    Therefore, we propose the following hierarchy of (generalized) contextuality: 
    \begin{quote}
        A GPT system $\B$ is said to be \textbf{at least as contextual} as a GPT system $\A$ if there exists a univalent simulation of $\A$ by a composite of $\B$ and a classical GPT system.
    \end{quote}
    It is the resource ordering of a resource theory with objects given by GPT systems and free operations given by univalent simulations with free access to classical systems. 
    In \cref{sec:contextualityerasure}, we discuss a possible interpretation of (non-univalent) simulations as expressing a particular kind of information erasure.

    \paragraph{Contextuality monotones.}
    As is common in resource theories, the hierarchy of contextuality is not a total order{\,---\,}there are GPT systems such that neither is at least as contextual as the other one.
    It is not even a partial order, because there exist distinct GPT systems which are equivalent.
    For example, this is the case for all noncontextual systems, such as the classical bit and the classical trit.
    Therefore, the hierarchy is given by a preorder and it cannot be fully represented by a single numerical value{\,---\,}the ``amount of contextuality"{\,---\,}assigned to each GPT system.
    However, in order to capture certain aspects of the hierarchy it is useful to define quantities which are order-preserving assignments of a number to each resource object.
    These are called \textit{resource monotones}. 
    
    We define a new contextuality monotone that we call \emph{classical excess}. 
    It expresses the minimal error of a univalent simulation of a given GPT system by any classical system. 
    In addition, we show that the optimal success probability for the parity-oblivious-multiplexing (POM) protocol \cite{Spekkens2008} with free classical systems is a monotone.

    \paragraph{Previous works on the matter.}
    Resource-theoretic perspective on contextuality has been developed in several works in the past.
    Most of these focus on the Kochen--Specker notion of contextuality \cite{KochenSpecker}. Despite being related to the generalized notion of contextuality (Kochen--Specker noncontextuality is the conjunction of measurement noncontextuality and outcome determinism for sharp measurements), Kochen--Specker contextuality favors certain frameworks that, instead, are not appropriate for developing a resource theory of generalized noncontextuality. 
    The work of Abramsky, Barbosa, and Mansfield \cite{abramsky2017contextual} uses a framework whose main objects are empirical models{\,---\,}tables of data, specifying probability distributions over the joint outcomes of sets of compatible measurements.  This framework is further developed in \cite{barbosa2023closing} and \cite{Karvonen2021} and is based on the sheaf theoretic approach to contextuality introduced in \cite{AB}. 
    Existing quantifications of Kochen--Specker contextuality are based on the memory cost \cite{Kleinmann2011}, the ratio of contextual assignments \cite{Svozil2012}, the relative entropy and contextual cost \cite{grudka2014quantifying}, the contextual robustness \cite{Li2020}, the contextual fraction~\cite{abramsky2017contextual}, and the rank of contextuality \cite{Horodecki2023}.
    In \cite{amaral2019resource} a review of several of the previous approaches towards a resource theory of Kochen--Specker contextuality is presented. 
    
    The first work on a resource theory of generalized contextuality in prepare-and-measure scenarios was presented by Duarte and Amaral in \cite{Duarte2018}. 
    They use the generalized-noncontextual polytope characterizing the contextual set of prepared-and-measured statistics defined in \cite{SchmidAll2018} to motivate the set of free operations and then define monotones based on known resource quantifiers for contextuality and nonlocality.
    As an application of such a resource theory, \cite{Wagner2023} uses it to simplify and robustify proofs of contextuality.

    \paragraph{On the use of general probabilistic theories.}

    Unlike Duarte and Amaral, we use the framework of GPTs. 
    A GPT system consists of a collection of states, effects and a probability assigned to each pair of a state and an effect.
    While noncontextuality is traditionally phrased in terms of operational theories, it can be extended to GPTs \cite{plavala2022incompatibility} and characterized by simplex embeddability \cite{Schmid2019}. 
    This condition states that a GPT system is noncontextual if and only if it can be embedded within a classical GPT system and thus characterizes the qualitative divide between contextual and noncontextual GPT systems.
    Our hierarchy of contextuality can be seen as a refinement thereof.
    For this purpose, we use univalent simulations between GPT systems, which are generalizations of noncontextual ontological models~\cite{MullerGarner2021}.

    \paragraph{Structure of the paper.} 
    In \cref{sec:prep_cont}, we recall the standard treatment of generalized noncontextuality in the framework of operational theories and ontological models. 
    In order to connect it to the language of GPTs (\cref{sec:GPT_systems}), we introduce the notion of an operational theory associated to a GPT system in \cref{sec:op_th_for_gpt}.Readers familiar with ontological models and GPTs may choose to skip directly to \cref{sec:GPT_simulations,sec:excess}, where we present (univalent) simulations of GPT systems and the excess measure, respectively.
    We then introduce our hierarchy of contextuality in \cref{sec:cont_preorder} and discuss its behaviour for composite systems in \cref{sec:contextuality_composite}.
    We prove that the classical excess is a resource monotone (\cref{sec:cont_monotone}) and construct a monotone from the parity-oblivious-multiplexing protocol in \cref{sec:POM}.
    In \cref{sec:contextualityerasure}, we discuss a possible interpretation of the non-free operations as involving information erasure, while in \cref{sec:previousworks} we discuss the relation to previous works on contextuality and GPTs, namely \cite{MullerGarner2021,Duarte2018,Schmid2019,GittonWoods2020,GittonWoods2022,shahandeh2021contextuality,selby2023accessible,Selby2024}.
    We conclude with a summary of the results and an outline of possible future research directions in \cref{sec:conclusion}.

\section{Operational theories, ontological models and  contextuality}\label{sec:prep_cont}

In this work we consider prepare-and-measure scenarios associated with a single system. An \textit{operational theory} associated with a prepare-and-measure scenario is defined by a list of possible preparations, measurements and the probabilities $\mathcal{P}(k|P,M)$ of obtaining the outcome $k$ of the measurement $M$ given that the system is prepared in the preparation $P$. 
An \textit{ontological model} of an operational theory provides a realist explanation of the operational predictions of the theory~\cite{Harrigan}. It does so by stipulating the existence of an \textit{ontic state space} for each given system, denoted with $\Lambda$, which is mathematically represented by a (standard Borel) measurable space. 
Each point $\lambda \in \Lambda$ represents an \textit{ontic state} that describes all the physical properties of the system. 
The ontological model associates each preparation $P$ in the operational theory with a conditional probability distribution $\mu_P(\lambda)\equiv \mu(\lambda|P) $ over ontic states. We call these \textit{epistemic states} as they represent states of knowledge about the underlying ontic states.  Each outcome $k$ given a measurement $M$, denoted by $[k|M]$, is associated with a conditional probability distribution $\xi(k|\lambda,M)$. The latter corresponds to the probability of obtaining outcome $k$ given that measurement $M$ is implemented on a system in the ontic state $\lambda$. 
An ontological model of an operational theory reproduces the predictions of the theory via the classical law of total probability, 
\begin{align} \label{opdata}
	\mathcal{P}(k|P,M)
	&=\sum_{\lambda\in\Lambda} \xi(k|\lambda,M)\mu(\lambda |P).
\end{align} 

An ontological model is preparation noncontextual if operationally equivalent preparation procedures are represented by identical probability distributions in the ontological model~\cite{Spekkens2005}. More formally, two preparation procedures $P$ and $P'$ are operationally equivalent if they provide the same operational statistics for all possible measurements, \textit{i.e.}, $\forall M: \mathcal{P}(k|P,M)= \mathcal{P}(k|P',M)$. In this case, we write $P\simeq P'$. An ontological model is preparation noncontextual if any two such preparations are represented by the same epistemic states:
\begin{equation}
	P\simeq P' \implies 
	\mu_P= \mu_{P'}.
\end{equation}

Similarly, two measurement outcomes $[k|M]$ and $[k'|M']$ are operationally equivalent if they give the same statistics for all possible preparations: $\P(k|M,P) = \P(k'|M',P)$ for all preparations $P$.  In this case, we write $[k|M]\simeq [k'|M']$. An ontological model is measurement noncontextual if any two such measurement outcomes are represented by the same response functions:
\begin{equation}
	[k|M]\simeq [k'|M'] \implies 
	\xi_{[k|M]}= \xi_{[k'|M']}.
\end{equation}

An operational theory is termed \textit{preparation noncontextual} (resp. \emph{measurement noncontextual}) if there exists a preparation noncontextual (resp. measurement noncontextual) ontological model for the theory, while it is termed  \textit{preparation contextual} (resp. \emph{measurement contextual}) if it does not admit of a preparation noncontextual  (resp. measurement noncontextual) ontological model. It was first proven in \cite{Spekkens2005} that quantum theory is preparation contextual and measurement noncontextual when outcome determinism is not assumed.  

We note that the notion of generalized contextuality can be also defined for transformations \cite{Spekkens2005}. 
In this work we are only concerned with prepare-and-measure scenarios and therefore do not consider transformation contextuality. 

\section{Contextuality for general probabilistic theories}
\label{sec:GPT_contextuality}
 
    One can often prove contextuality of a theory by studying statistical behaviours of an individual system $\A$ (such as a qubit) in prepare-and-measure scenarios.
    The information that suffices to describe the possible behaviours constitutes of
    \begin{itemize}
        \item a set $\Omega$ of possible states that $\A$ can be prepared in, according to the theory in question,
        
        \item a set of possible measurements that can be applied to $\A$, each consisting of a collection of effects associated with the measurement outcomes, and
        
        \item for every state $\omega \in \Omega$ and every effect $e$ the probability of obtaining $e$ when measurement $M$ is applied to system $\A$ prepared in state $\omega$.

    \end{itemize}
    This information is commonly expressed as a system in a general probabilistic theory{\,---\,}a GPT system.
    
In the following we omit any mention of dynamics that the system may undergo.
     This simplification has no consequence on our discussion of preparation and measurement contextuality.
    
    \subsection{GPT systems}\label{sec:GPT_systems}
    
    Let us now discuss the mathematical description of GPT systems~\cite{Mackey_mathematical_1963,Ludwig_versuch_1964,davies_operational_1970,ludwig_axiomatic_1985,lami2018nonclassical,Beneduci_2022} and the definition of contextuality in this context.

    In the following definition, $V_{\A}$ is a normed vector space and $V_{\A}^*$ is the topological dual of $V_{\A}$. 
    The canonical pairing function is denoted by $  \ph \cdot \ph : V_{\A}^* \times V_{\A} \to \mathbb{R}$, i.e.\ $e \cdot \omega \coloneqq e(\omega)$.

       \begin{definition}\label{def:GPT_system}
        A \newterm{GPT system} $\A$ is specified by two non-empty convex sets
        \begin{equation}
            \Omega_{\A} \subseteq V_{\A},  \qquad E_{\A} \subseteq V_{\A}^*
        \end{equation}
        that are called the \newterm{state space} and \newterm{effect space}, respectively, while their elements are states and effects of $\A$.
        We require that $E_\A$ contains the \newterm{null effect} $0_\A$ and the \newterm{unit effect} $1_\A$ satisfying
        \begin{equation}\label{eq:zero_unit_effects}
            0_\A \cdot v = 0 \quad \forall \, v \in V_\A  \qquad \text{and} \qquad  1_\A \cdot \omega = 1 \quad \forall \, \omega \in \Omega_\A
        \end{equation}
        respectively.
        In order for the pairing to have a probabilistic interpretation, we also require $e \cdot \omega \in [0,1]$ for all effects $e$ and all states $\omega$.
        Finally, to avoid complications arising from degenerate systems, we require that  
        \begin{equation}
            V_\A = \mathrm{span}(\Omega_\A)  \qquad \text{and} \qquad V_\A^* = \mathrm{span}(E_\A)
        \end{equation}
        hold, so that states can be distinguished by effects:
        \begin{equation}
            \label{eq:non-degenerate}
            e \cdot \omega = e \cdot \omega' \quad \forall \, e \in E_A  \qquad \implies \qquad  \omega = \omega',
        \end{equation}
        and similarly that effects can be distinguished by states: 
        \begin{equation}
            \label{eq:non-degenerate_2}
            e \cdot \omega = e' \cdot \omega \quad \forall \, \omega \in \Omega_A  \qquad \implies \qquad  e = e'.
        \end{equation}
    \end{definition}

    The set of all GPT systems is denoted by $\G$ and the subset of finite-dimensional ones by $\G_\fin$.

    A GPT system encodes all the information about the statistical behaviour of a physical system in prepare-and-measure scenarios.
    Let us give a few examples of GPT systems.

    \begin{example}[Finite-dimensional quantum system as a GPT system]\label{ex:quantum_GPT}
        To each Hilbert space $\Cl^n$ with $n \in \Nl$, one can associate a quantum GPT system $\Q_n$.
        The real vector space $V_{\Q_n} \simeq \R^{n^2}$ is the space of Hermitian operators on the underlying Hilbert space.
        The state space is the set of density operators and the effect space is the set of quantum effects, which are positive semi-definite operators $e$ that satisfy
        \begin{equation}
            0 \leq e \leq 1,
        \end{equation}
        where $1$ is the identity operator.
        The pairing is nothing but the Hilbert--Schmidt inner product:
        \begin{equation}
            e \cdot \omega = \tr (e^\dagger \omega).
        \end{equation}
    \end{example}

    Likewise, every classical probabilistic system is a GPT system~\cite[\S 4]{davies_operational_1970}.

    \begin{example}[Finite classical GPT systems]\label{ex:classical_GPT}
    We denote the finite $n$-level classical system (as well as its state space) by $\Delta_n$, with underlying vector space $\mathbb{R}^n$ and state space given by the simplex
        \begin{equation}
            \Delta_n \coloneqq \conv \Set*[\big]{ \delta_i  \given  i \in \{1, \ldots, n\} },
        \end{equation}
        where $\{\delta_i\}_i$ is a chosen orthonormal basis of $\mathbb{R}^n$ and $\conv$ denotes the convex hull operation.
        The effect space consists of all the linear functionals $\xi$ that satisfy 
        \begin{equation}
            \xi \cdot \mu \in [0,1]
        \end{equation}
        for every distribution $\mu \in \Delta_n$.
        If we think of $\mu$ as a column vector, then the possible effects are all the associated row vectors with entries in the unit interval $[0,1]$.
        In other words, the effect space is 
        \begin{equation}
            \Delta^*_n \coloneqq \conv \Set*[\big]{ \upvarrho_\alpha  \given  \alpha \in \{0,1\}^n },
        \end{equation}
        where
        \begin{equation}
            \upvarrho_\alpha \coloneqq \sum_{i=1}^n \alpha_i \delta^*_i
        \end{equation}
        and $\{\delta^*_i\}_i$ is the dual basis to the one above.
    \end{example}
    
    Besides finite classical systems, we will also need the notion of a countably infinite one.

    \begin{example}[Countable classical GPT system]
        \label{ex:integer_simplex}
       The countable classical system $\Delta_{\N}$ has associated vector space  $\ell^1$, the Banach space of sequences whose series are absolutely convergent.
       The state space, also denoted by $\Delta_{\N}$, consists of all probability measures on $\N$.
       The effect space $\Delta_\N^*$ consists of arbitrary sequences in $[0,1]^\N$ and forms a convex subset of $\ell^\infty$, the space of bounded sequences, which is the dual of $\ell^1$. 
    \end{example}

    Besides individual systems, we will occasionally also need to refer to composite ones.
In general, the composite of two GPT systems is not unique{\,---\,}see \cite[Section 5]{plavala2021general} for an in-depth discussion of tensor products of GPT systems.
    However, among all the possible choices, there is a `minimal composite' of two GPT systems $\A$ and $\B$, whose state and effect space merely contain the separable ones. 
    Any meaningful composite of two GPT systems necessarily contains their minimal composite as a subsystem.

    \begin{definition}\label{def:minimal_tensor}
        Given two GPT systems $\A = (\Omega_\A, E_\A, V_\A)$ and $\B = (\Omega_\B, E_\B,V_\B)$ 
 we define the \newterm{minimal composite} system $\A \otimes \B$ to be given by the separable states and separable effects on $\A$ and $\B$. Formally  $\A \otimes \B\coloneqq (\Omega_{\A \otimes \B} , E_{\A \otimes \B}, V_{\A \otimes \B})$ where,
        \begin{align}
            V_{\A \otimes \B} &\coloneqq V_\A \otimes V_\B, \\
            \Omega_{\A \otimes \B} &\coloneqq  \conv \Set*[\big]{ \omega_\A \otimes \omega_\B  \given  \omega_\A \in \Omega_{\A} \text{ and } \omega_\B \in \Omega_{\B} },\\
            E_{\A \otimes \B} &\coloneqq  \conv \Set*[\big]{ e_\A \otimes e_\B  \given  e_\A \in E_{\A} \text{ and } e_\B \in E_{\B} },
        \end{align}
    \end{definition}

   Of particular interest is the tensor product of an arbitrary GPT system and a classical GPT system, in which case there is a unique choice of a sensible tensor product corresponding to the minimal composite of GPT systems (see~\cref{sec:contextuality_composite} for further discussion about this point). 

\subsection{Operational theory associated to a GPT system}\label{sec:op_th_for_gpt}

    As we articulate in \cref{sec:prep_cont}, proofs of (preparation) contextuality involve \emph{distinct} preparation procedures, which are nevertheless operationally \emph{equivalent}.
    Since any two distinct states of a GPT system are operationally \emph{inequivalent}, we cannot identify GPT states with preparation procedures.
    Instead, we think of a preparation procedure as an ensemble of GPT states, each occurring with a specified probability.
    Similarly, proofs of measurement contextuality involve \emph{distinct} measurement outcomes, which are operationally \emph{equivalent}.
    Since any two distinct effects of a GPT system are operationally \emph{inequivalent}, we cannot identify GPT effects with measurement outcomes.
    Instead, we think of a measurement procedure as a collection of GPT effects, each of which is a possible outcome of the given measurement.
    See \cite[Section 3]{plavala2022incompatibility} for more details.

    Such preparations and measurement outcomes then provide a notion of a canonical operational theory associated to a GPT system,\footnote{Conversely, one can always obtain a GPT system from an operational theory.
    Specifically, each state is an equivalence class of preparation procedures with respect to operational equivalence, and similarly one can define the effects \cite{holevo_probabilistic_2011}.} to which we can assign ontological models as in \cref{sec:prep_cont}.
    In this way, we obtain a canonical notion of an ontological model for GPT systems and a corresponding notion of noncontextuality \cite[Definitions 2 and 3]{plavala2022incompatibility}.

\subsection{Simulations of GPT systems}\label{sec:GPT_simulations}

    We are, however, interested in comparing GPT systems and not just in distinguishing the noncontextual ones.
    For this purpose, we use the notion of (univalent) simulations.
    
    \begin{definition}[{\cite[Definition 1]{MullerGarner2021}}]
        \label{def:simulation}
        Given two GPT systems $\A$ and $\B$, an $\epsilon$-\newterm{simulation} of $\A$ by $\B$ is a pair of 
        \begin{itemize}
            \item a multi-valued function $\Gamma : \Omega_\A \to \Omega_\B$ called \newterm{state $\epsilon$-simulation} and 
            \item a multi-valued function $\Theta : E_\A \to E_\B$ called \newterm{effect $\epsilon$-simulation},
        \end{itemize}
        which is empirically adequate up to an error $\epsilon \in [0,\infty)$:
        \begin{equation}\label{eq:sim_emp_adeq}
            \abs*[\big]{e \cdot \omega - f \cdot \gamma}  \leq \epsilon  \qquad  \forall \, \omega \in \Omega_\A ,\; \forall \, e \in E_\A , \; \forall \, \gamma \in \Gamma(\omega), \; \forall \, f \in \Theta(e) ,
        \end{equation}
        preserves convex mixtures:\footnotemark{}
        \begin{equation}\label{eq:sim_conv_mix}
            \lambda \Gamma(\omega) + (1 - \lambda) \Gamma(\omega')  \subseteq  \Gamma \bigl( \lambda \omega + (1 - \lambda) \omega' \bigr)  \qquad  \text{ for all } \lambda \in [0,1] \text{ and } \omega, \omega' \in \Omega_\A,
        \end{equation}
        and also preserves the null effect, i.e.\ we have $0_\B \in \Theta(0_\A)$.
        \footnotetext{Analogously, we require that the same holds for mixtures of effects rather than states, just as in \cite{MullerGarner2021}.}%
        
        A $0$-simulation is called an \newterm{exact} simulation.\footnote{Sometimes, we refer to an exact simulation also just as \emph{a simulation} for brevity.}
        If both $\Gamma$ and $\Theta$ are single-valued, we say that the simulation is \newterm{univalent}.
    \end{definition}
    
    One of the immediate consequences of the definition is that the images of states  under a state simulation (and of effects under an effect simulation) are non-overlapping. 
    Indeed, by empirical adequacy and the fact that $\Theta(e)$ is always a non-empty subset of $E_\B$, we find 
    \begin{equation}\label{eq:image_nonoverlap}
        \begin{split}
            \gamma \in \Gamma(\omega_1) \text{ and } \gamma \in \Gamma(\omega_2)  &\implies  \forall e \in E_\A \; : \; e \cdot \omega_1 = f \cdot \gamma = e \cdot \omega_2 \\
            &\implies \omega_1 = \omega_2,
        \end{split}
    \end{equation}
    where the second implication is by \eqref{eq:non-degenerate}.

    \begin{definition}\label{def:embed_preoder}
        Given a pair of GPT systems $(\A, \B)$, we say that $\A$ is \newterm{embeddable} within $\B$, denoted by $\A \hookrightarrow \B$, if there exists an exact univalent simulation of type $\A \to \B$.
        The relation $\hookrightarrow$ is called the \newterm{embeddability preorder}.
    \end{definition}
    Indeed, one can easily show that $\hookrightarrow$ is a preorder relation.
    Specifically, identity maps provide univalent simulations of type $\A \hookrightarrow \A$, which proves the reflexivity of $\hookrightarrow$.
    To prove that it is a transitive relation, one can use the triangle inequality (see \cref{prop:excess_triangle}).

    It has also been shown that
    \begin{itemize}
        \item a discrete ontological model of a GPT system $\A$ provides an exact simulation of $\A$ by a finite classical GPT system (and vice versa) \cite[Theorem 1]{MullerGarner2021}, 

        \item a discrete \emph{noncontextual} ontological model of a GPT system $\A$ is the same as an exact \emph{univalent} simulation of $\A$ by a finite classical GPT system \cite[Corollary 1]{MullerGarner2021}.
    \end{itemize}
    In this sense (univalent) simulations generalize (noncontextual) ontological models.

    We note that continuous ontological models of GPT systems correspond to simulations by infinite dimensional systems. 
    However, for the purposes of this work and in keeping with the related approaches of~\cite{Schmid2019,MullerGarner2021}, we restrict ourselves to discrete ontological models.

    \begin{example}[Simulations of a bit by a trit]\label{ex:bit_trit}
        Given a bit $\Delta_2$ and a trit $\Delta_3$ one can define a univalent simulation $(\Gamma, \Theta)$ of $\Delta_2$ by $\Delta_3$ via a direct embedding. 
        Namely, the state and effect simulations act on the bases of $\mathbb{R}^2$ and its dual (see \cref{ex:classical_GPT}) via
        \begin{equation}
            \Gamma \left( \delta_i \right) = \left\{ \delta_i \right\},  \qquad  \Theta \left( \delta_i^* \right) = \left\{\delta_i^*\right\}
        \end{equation}
        and can be extended to all states and effects by linearity.
        This simulation is univalent since both $\Gamma$ and $\Theta$ are single-valued.
        
        One can also adapt it to a simulation that is not preparation univalent while still being measurement univalent.
        For instance, we can keep the action on $\delta_1$, $\delta_1^*$ and adjust the simulation of the other deterministic state and effect to be
        \begin{equation}
            \Gamma \left( \delta_2 \right) = \Set*[\big]{ \lambda \delta_2 + (1-\lambda) \delta_3 \given \lambda \in [0,1] },  \qquad  \Theta \left( \delta_2^* \right) = \left\{\delta_2^* + \delta_3^*\right\}
        \end{equation}
        Once we require that $\Theta$ also preserves the unit effect, we can extend both to $\Delta_2$ and $\Delta_2^*$ by convex-linearity (and using Minkowski sum for the state simulation).
    \end{example}
\begin{example}[Holevo-Beltrametti-Bugasjki model~\cite{Beltrametti,holevo_probabilistic_2011}]\label{ex:HBB}
    Given a GPT system $\A = (\Omega_\A, E_\A, V_\A)$ with $n \in \mathbb{N}$ extremal states, the Holevo-Beltrametti-Bugasjki model is given by the following simulation $(\Gamma,\Theta) : \A \to \Delta_{n}$ where
    \begin{equation}
        \Gamma \left( \omega_i \right) = \left\{ \delta_i \right\},  \qquad  \Theta \left( e \right) = \left\{ \sum_{i = 1}^n (e \cdot \omega_i ) \delta_i^*\right\}
    \end{equation}
    where $\omega_i$ is the $i$-th extremal state of $\Omega_\A$ and $e$ is an arbitrary effect in $E_\A$.
    However, it is not a univalent simulation unless $\Omega_\A$ is a simplex.
    While the effect simulation is clearly single-valued, to extend the state simulation to all states of $\Omega_\A$ we need to ensure that property \eqref{eq:sim_conv_mix} holds.
    That is, for a given state $\omega$, its simulation $\Gamma(\omega)$ is given by the set of all convex decompositions of $\omega$ into extremal states (and applying $\Gamma$ to each extremal state to land within $\Delta_n$).
\end{example}

\subsection{Properties of univalent simulations}\label{sec:excess}

    The following proposition, which stresses how the composition with a classical system does not affect the univalency of a simulation, is relevant for motivating our notion of contextuality preorder further on (\cref{def:cont_preorder}) as well as for quantifying contextuality. 

    \begin{proposition}
        \label{prop:emb_reduction}
        Let $n$ and $k$ be arbitrary natural numbers.
        If there is a univalent $\epsilon$-simulation $(\Gamma, \Theta) : \A \otimes \Delta_n \to \B$, then there is also a univalent $\epsilon$-simulation $\A \to \B \otimes \Delta_k$.
    \end{proposition}
    \begin{proof}
        We show this for $\epsilon = 0$, the general case is analogous.
        In particular, we define maps $\Gamma_1 : \Omega_\A \to \Omega_\B \otimes \Delta_k$ and $\Theta_1 : E_\A \to E_\B \otimes \Delta_k^*$ via
        \begin{equation}
            \Gamma_1 (\omega) \coloneqq \Gamma( \omega \otimes \delta_1) \otimes \delta_1,  \qquad \qquad  \Theta_1 (e) \coloneqq \Theta( e \otimes \delta^*_1) \otimes \delta^*_1.
        \end{equation}
        It follows that this is an exact univalent simulation by the assumption that $(\Gamma, \Theta)$ is.
        Specifically, we have
        \begin{equation}
            \Theta( e \otimes \delta^*_j) \cdot \Gamma( \omega \otimes \delta_i ) = (e \cdot \omega ) ( \delta^*_j \cdot \delta_i) = 
            \begin{dcases}
                e \cdot \omega  & \text{if } i=j \\
                0               & \text{otherwise,}
            \end{dcases}
        \end{equation}
        which implies $\Theta_1 (e) \cdot \Gamma_1 (\omega) = e \cdot \omega$.
    \end{proof}

    By determining the smallest error $\epsilon$ for which there is a univalent $\epsilon$-simulation of type $\A \to \B$, we obtain a meaningful notion of how far $\A$ is from being embeddable within $\B$.

    \begin{definition}\label{def:eps_min}
        The \newterm{excess} is the function $\excess : \G \times \G \to \mathbb{R}$ defined by
        \begin{equation}
            \excess(A,B) \coloneqq \inf \Set*[\big]{ \epsilon \in [0,\infty) \given  \text{there is a univalent $\epsilon$-simulation $A \to B$} },
        \end{equation}
        We call $\excess(A,B)$ the $A$-excess within $B$.
    \end{definition}

    Note that the value of excess is always in the interval $[0,1]$.
    One can show that it cannot exceed $1$ by observing that there is always a univalent $1$-simulation of type $\A \to \Delta_1$ and a univalent $0$-simulation of type $\Delta_1 \to \B$, and using the following proposition.

    \begin{proposition}
        \label{prop:excess_triangle}
        For any three GPT systems $A$, $B$, and $C$, we have the \newterm{triangle inequality for excess}
        \begin{equation}
            \label{eq:triangle_ineq}
            \excess(A, C) \leq \excess(A,B) + \excess(B,C).
        \end{equation}
    \end{proposition}
    \begin{proof}
        Consider an arbitrary univalent $\epsilon$-simulation $(\Gamma,\Theta) : A \to B$ and an arbitrary univalent $\gamma$-simulation $(\alpha,\beta) : B \to C$.
        We prove the statement by showing that $(\alpha \, \Gamma, \beta \, \Theta) $ is a univalent simulation of type $A \to C$ with error $\epsilon + \gamma$, a result stated as Lemma 3 in \cite{MullerGarner2021}.

        To show that the composite simulation is empirically adequate as expressed by Inequality \eqref{eq:sim_emp_adeq}, we can use the triangle inequality.
        In particular, for all $e_\A \in E_A$ and all $\omega_\A \in \Omega_\A$, we have
       
        \begin{equation}
            \begin{aligned}
                \abs*[\big]{e_\A \cdot \omega_\A &- \beta \, \Theta(e_\A) \cdot \alpha \, \Gamma(\omega_\A) } \\
                &\leq \abs*[\big]{e_\A \cdot \omega_\A - \Theta(e_\A) \cdot \Gamma(\omega_\A) } + \abs*[\big]{\Theta(e_\A) \cdot \Gamma(\omega_\A) - \beta \, \Theta(e_\A) \cdot \alpha \, \Gamma(\omega_\A) } \\
                &\leq \epsilon + \gamma,
            \end{aligned}
        \end{equation}
        where the second inequality follows because $\Theta(e_\A)$ is an effect of $B$ and $\Gamma(\omega_\A)$ is a state of $B$, by assumption.
    \end{proof}

    Note that whenever $\B$ is embeddable within $\C$, we have $\excess(\B,\C) = 0$ and the triangle inequality reads
    \begin{equation}
        \excess(\A, \C) \leq \excess(\A,\B)
    \end{equation}
    for any GPT system $\A$.
    A similar statement follows by setting the first term on the right-hand side of inequality \eqref{eq:triangle_ineq} to zero.
    These two statements can be summarized as the following corollary of \cref{prop:excess_triangle}.
    
    \begin{corollary}
        \label{cor:excess_monotone}
        Let $\A$ and $\C$ be arbitrary GPT systems.
        With respect to the embeddability preorder among GPT systems, we have the following order-preserving functions:
        \begin{align}
            \excess(\A, \ph) &: (\G, \hookrightarrow) \to (\mathbb{R},  \geq), \\
            \label{eq:excess_monotone2}
            \excess(\ph, \C) &: (\G,  \hookrightarrow) \to (\mathbb{R},  \leq).
        \end{align}
    \end{corollary}
    
    In other words, the monotonicity of $\excess(\ph, \C)$ from \eqref{eq:excess_monotone2} says that if 
    \begin{equation}\label{eq:excess_ineq}
        \excess(\A, \C) > \excess(\B,\C)
    \end{equation}
    holds (i.e.\ the minimal error of simulation by $C$ is strictly higher for $\A$ than for $\B$), then there cannot be an exact univalent simulation of type $\A \to \B$.
    We can thus use the excess within $\C$, for any GPT system $\C$, as a witness of the impossibility to construct an exact univalent simulation.
    
    While simulations are crucial to express a binary division between noncontextual and contextual GPT systems, the embeddability preorder $\hookrightarrow$ is \textit{not} a satisfactory refinement of this binary division. 
    The simplest argument to see this is that all classical GPT systems (which are trivially noncontextual) should be equivalent in a resource theory of contextuality. 
    However, for any integer $m$ greater than $k$, the classical GPT system $\Delta_m$ \emph{is not} embeddable within $\Delta_k$.
    Therefore, we cannot think of the embeddability preorder $\hookrightarrow$ as a resource ordering of some resource theory of contextuality.
    In the following section, we motivate a preorder that further refines the embeddability relation $\hookrightarrow$ and addresses the issue above by stipulating classical GPT systems to be `free resources'.
   
\section{Hierarchy of contextuality}
\label{sec:RT_contextuality}

    Recall that our interest in a hierarchy of contextuality of GPT systems stems from the wish to measure how useful contextuality is as a resource in information-theoretic applications.
    How can one decide whether a given measure of interest is indeed a well-motivated \emph{measure of contextuality}?
    The resource-theoretic perspective \cite{coecke2016mathematical,gonda2021resource}, fruitfully applied to studies of entanglement and many other resources in the past, suggests to first motivate a preorder $\succeq$ among the objects of interest.
    Importantly, some objects may be incomparable according to $\succeq$.
    In our case, the resource objects are GPT systems and the resource ordering is the advertised hierarchy of contextuality.
    The interpretation of such an ordering is that the relation $\A \succeq \B$ expresses the statement `$\A$ is at least as contextual as $\B$ is'.
    Given a good motivation for such a hierarchy, one can answer the above question as follows:
    \begin{quote}
        A (real-valued) function on the set $\G$ of all GPT systems is a measure of contextuality if and only if it is order-preserving with respect to the hierarchy of contextuality $\succeq$.
    \end{quote}

\subsection{Motivation and the resource theory}
\label{sec:cont_preorder}

Let us now list and discuss our desiderata for a sensible hierarchy of  contextuality among GPT systems.
They are:
\begin{enumerate}
    \item \label{it:classical_free} \textbf{Access to classical GPT systems is free.} That is, we have
    \begin{equation}
        \label{eq:classical_free}
        \A \preceq \A \otimes \Delta_n  \quad \text{ and } \quad \A \otimes \Delta_n \preceq \A
    \end{equation}
    for every GPT system $\A$ and every classical GPT system $\Delta_n$. 
    Here, $\otimes$ denotes the minimal composite from \cref{def:minimal_tensor}.

    \item \label{it:sim_free} \textbf{Exact univalent simulations are free.} That is, if $\A$ is embeddable within $\B$, then we have $\A \preceq \B$.
\end{enumerate}

Notice that the first condition in \cref{eq:classical_free}, i.e.\ $\A \preceq \A \otimes \Delta_n $, follows from desideratum \ref{it:sim_free} because every GPT system $\A$ is embeddable within $\A \otimes \Delta_n$ (see \cref{prop:emb_reduction}).

There are several consequences of desiderata 1 and 2 which we want to point the reader's attention to, before we discuss our motivation to introduce these desiderata.
\begin{enumerate}[(i)]

    \item Desideratum \ref{it:sim_free} implies that the trivial GPT system $\Delta_1$ is \textit{least contextual} in the sense that it is the bottom element of $\preceq$. 
    That is, we have $ \Delta_1 \preceq \A$ for every GPT system $\A$.

    \item Consequently, using the fact that $\Delta_1 \otimes \Delta_n$ is the same GPT system as $\Delta_n$, every classical GPT system is also at the bottom by Desideratum \ref{it:classical_free}. 

    \item Finally, using the fact that noncontextual systems are simplex embeddable and Desideratum~\ref{it:sim_free}, we conclude that all noncontextual GPT systems are equivalent to each other and lie at the bottom of the hierarchy.

\end{enumerate}

\paragraph{Motivation for Desideratum \ref{it:classical_free}.} 
Contextuality, as we are thinking of it, is a property that expresses the degree to which a system's behaviour escapes a purely classical explanation.
In this sense, considering additional classical systems should not affect this property.
An important point to note is that within the tensor product $\A \otimes \Delta_n$ (\cref{def:minimal_tensor}), only correlations involving separable states between the generic GPT system $\A$ and the classical system are allowed.
We can think of behaviours of $\A \otimes \Delta_n$ as dilations of those of $\A$ \cite{houghton2021mathematical}, given a classical environment variable $\Delta_n$.

If there were classical environmental variables, whose knowledge reduces the contextuality present in $\A$, then the contextuality of $\A$ would not be an authentic feature of $\A$, rather a consequence of one's too narrow focus on $\A$ as opposed to its dilation $\A \otimes \Delta_n$.
We can also support this argument on the technical side.
Indeed, by \cref{prop:emb_reduction}, any noncontextual model of $\A \otimes \Delta_n$ provides also a noncontextual model of $\A$.
In other words, $\A$ is noncontextual if and only if $\A \otimes \Delta_n$ is.\footnote{Note that under the minimal tensor product of~\Cref{def:minimal_tensor} allowing access to arbitrary noncontextual systems (and not just classical systems) would lead to the same hierarchy, which follows from the fact that access to classical systems already places all noncontextual systems at the bottom of the hierarchy.}
The hierarchy of contextuality ought to capture this fact.

\paragraph{Motivation for Desideratum \ref{it:sim_free}.}
We can also think of contextuality (of a GPT system)  as the inability to provide a noncontextual model (i.e.\ a univalent simulation by a classical system) for its statistical behaviours. 
If a GPT system $\A$ can be exactly simulated by a GPT system $\B$ via a univalent simulation, then any obstruction to constructing such a noncontextual model for $\A$ must already be present in $\B$.
In this sense, we think of $\B$ as being \qmarks{at least as contextual} as $\A$ is.
Indeed, provided with a noncontextual model of $\B$, i.e.\ an exact univalent simulation of $\B$ by a classical GPT system, one can construct a noncontextual model of $\A$ by composition with the simulation of $\A$ by $\B$.

The preorder relation among finite-dimensional GPT systems in $\G_\fin$ that satisfies precisely these requirements is the following one.
\begin{definition}
    \label{def:cont_preorder}
    We say that $\B$ is at least as contextual as $\A$, denoted $\A \preceq \B$, whenever there is an exact univalent simulation of $\A$ by $\B \otimes \Delta_n$ for some finite-dimensional classical GPT system $\Delta_n$.
    The preordered set $(\G_\fin,\contord)$ is called the \newterm{hierarchy of contextuality}.
\end{definition}
In other words, we define the hierarchy via embeddability given an additional classical system:
\begin{equation}
    \A \preceq \B  \quad \iff \quad  \exists \, n \in \mathbb{N}   \text{ such that } \A \hookrightarrow \B \otimes \Delta_n \text{ holds.}
\end{equation}
To show that the hierarchy $\preceq$ is a preorder relation, we can use a resource-theoretic perspective, i.e.\ we can show that the set of free operations in the following resource theory is closed under composition.
\begin{definition}\label{def:RT_contextuality}
    The \newterm{resource theory of GPT-contextuality} is defined as follows.
    Its objects are finite-dimensional GPT systems in $\G_\fin$.
    The set of transformation from $\B$ to $\A$ is identified with exact simulations of type ${\A \to \B \otimes \Delta_n}$.
    Two such simulations 
    \begin{equation}\label{eq:res_trans}
        (\Gamma, \Theta) : \A \to \B \otimes \Delta_m  \quad \text{and} \quad 
        (\alpha, \beta) : \B \to \C \otimes \Delta_k
    \end{equation}
    can be composed (sequentially) to produce a simulation
    \begin{equation}
        (\alpha \, \Gamma, \beta \, \Theta ) : \A \to \C \otimes \Delta_{m \cdot k}
    \end{equation}
    via the canonical isomorphism between $\Delta_m \otimes \Delta_k$ and $\Delta_{m \cdot k}$.
    Free transformations are those simulations that are also \emph{univalent}.
\end{definition}
This gives rise to a concrete resource theory in the sense of \cite[Definition 3.10]{gonda2021resource}.
Moreover, the resulting resource ordering coincides with our hierarchy of contextuality. 
In a general resource theory, the existence of a free transformation from a resource object $r$ to a resource object $s$ defines the resource ordering denoted by $s \preceq r$, i.e.\ $r$ is at least as good as a resource than $s$ is.

In our work, a simulation of type $\C \to \D$ is interpreted as providing effective access to system $\C$ to an agent with access to $\D$. 
In this sense, access to $\D$ grants the agent at least as many capabilities as access to $\C$.
For this reason, a simulation of type $\A \to \B \otimes \Delta_n$ is identified with a resource-theoretic transformation from $\B$ to $\A$ and not vice versa.
Note, however, that it is not in general realisable by a physical transformation from $\B$ to $\A$ in the GPT sense.
A physical transformation $\B \to \A$ is given by a linear map $M : V_\B \to V_\A$ such that $M(\Omega_\B) \subseteq \Omega_\A$ and $M^*(E_\A) \subseteq E_\B$. 
For further discussion see~\Cref{sec:contextualityerasure,app:phys_sim}.

\subsection{Contextuality of composite systems} \label{sec:contextuality_composite}
    
    When studying a resource theory, it is often also important to specify how resource objects are composed to obtain a joint resource.
    As discussed in \cref{sec:GPT_systems}, GPT systems do not have a canonical tensor product. 
    However, any sensible choice of a tensor product $\tildetensor$ should give rise to a symmetric monoidal category with exact simulations as morphisms \cite[Section 5.2]{plavala2021general}.
    Since the transformations in our resource theory of GPT-contextuality are not merely simulations of GPT systems, but also involve additional classical systems, it takes a bit more work to show that they (and the free transformations) are closed under parallel composition.
    We prove this in the rest of this section for the case of free operations.
    
    Consider two exact univalent simulations of types $\A_1 \to \B_1 \otimes \Delta_m$ and $\A_2 \to \B_2 \otimes \Delta_k$ respectively.
    The assumption that $\tildetensor$ is the monoidal product of a symmetric monoidal category then means that there is also an exact univalent simulation of type
    \begin{equation}\label{eq:tensor_simulation}
        \A_1 \tildetensor \A_2 \to \left( \B_1 \otimes \Delta_m \right) \tildetensor \left( \B_2 \otimes \Delta_k \right).
    \end{equation} 
    Assuming that there is a unique composite of any GPT system with a classical one given by the canonical composite of~\Cref{def:minimal_tensor} (see~\cite{Aubrun_2021,aubrin_entanglement_2022} for a derivation of the uniqueness of the composition with a classical system under the assumption of local tomography and~\cite{Erba2020} for an alternative composition of classical systems without this assumption) and that $\tildetensor$ is associative, the codomain can be expressed as
    \begin{equation}\label{eq:classical_tensor_move}
        \begin{split}
            \left( \B_1 \otimes \Delta_m \right) \tildetensor \left( \B_2 \otimes \Delta_k \right) &=  \left( \B_1 \tildetensor \Delta_m \right) \tildetensor \left( \B_2 \tildetensor \Delta_k \right) \\
            &\cong \left( \B_1 \tildetensor \B_2 \right) \tildetensor \left( \Delta_m \tildetensor \Delta_k \right) \\
            &\cong \left( \B_1 \tildetensor \B_2 \right) \tildetensor \Delta_{m \cdot k} \\
            &= \left( \B_1 \tildetensor \B_2 \right) \otimes \Delta_{m \cdot k},
        \end{split}
    \end{equation}
    where the third equality uses the standard isomorphism between $\Delta_m \otimes \Delta_k$ and $\Delta_{m \cdot k}$. 
    In the first and final steps we also use the assumption that composition with a classical system is unique.
    Consequently, we get a parallel composite of the two free transformations as an exact univalent simulation of type 
    \begin{equation}
        \A_1 \tildetensor \A_2 \to \left( \B_1 \tildetensor \B_2 \right) \otimes \Delta_{m \cdot k}.
    \end{equation}
    
    As a consequence, the parallel composition of resources (i.e.\ GPT systems) via $\tildetensor$ respects the hierarchy of contextuality in the sense that we have
    \begin{equation}\label{eq:tensor_order}
        \A_1 \preceq \B_1  \; \land \;  \A_2 \preceq \B_2  \quad \implies \quad  \A_1 \tildetensor \A_2 \preceq \B_1 \tildetensor \B_2,
    \end{equation}
    for arbitrary GPT systems $\A_i$ and $\B_i$.
    Condition \eqref{eq:tensor_order} is part of the definition of a resource theory via ordered commutative monoids in \cite{Fritz2017}. 
    It means that we can consistently apply the hierarchy of contextuality also to express the contextuality of composite GPT systems, as long as their composition satisfies a few basic properties as described above.

\subsection{Quantifying contextuality via the classical excess}
\label{sec:cont_monotone}

As is common in resource theories \cite{gonda2023monotones}, we can use monotones{\,\textemdash\,}i.e.\ order-preserving functions from resources to numbers{\,\textemdash\,}to study the hierarchy of contextuality.
On the one hand, these can provide a lens through which to extract properties of the preorder.
On the other hand, they can be used as quantitative measures of contextuality.
The latter perspective is particularly useful when the monotones come equipped with an (operational) interpretation that allows one to identify \textit{which aspect of contextuality} they are measuring.

In order to construct a specific contextuality monotone, we will make use of the excess function from \cref{def:eps_min}.
Of particular interest is $\excess(\A, \Delta_{\mathbb{N}})$, the $\A$-excess within the countable classical GPT system (\cref{ex:integer_simplex}).
For brevity, we call it the \newterm{classical excess} of $\A$.  
A crucial ingredient in the following proof is that $\Delta_{\mathbb{N}}$ is isomorphic to $\Delta_{\mathbb{N}} \otimes \Delta_n$ for any finite $n$.

\begin{lemma}
    \label{lem:clas_no_excess}
    For any GPT system $\A$ and any finite-dimensional classical GPT system $\Delta_n$, we have
    \begin{equation}
        \excess(\A \otimes \Delta_n,\Delta_\N) = \excess(\A,\Delta_\N)
    \end{equation}
\end{lemma}
\begin{proof}
    Firstly, the inequality 
    \begin{equation}
        \excess(\A \otimes \Delta_n,\Delta_\N) \geq \excess(\A,\Delta_\N)
    \end{equation}
    is a direct consequence of \cref{prop:emb_reduction} and the definition of excess.

   For the converse inequality, note that any univalent $\epsilon$-simulation $(\Gamma,\Theta) : \A \to \Delta_\N$ can be canonically extended to a univalent $\epsilon$-simulation
    \begin{equation}
        \bigl( \Gamma \otimes \id, \Theta \otimes \id \bigr) \; : \; \A \otimes \Delta_n \to \Delta_\N \otimes \Delta_n,
    \end{equation}
    so that we get 
    \begin{equation}\label{eq:emb_extension}
        \excess(\A,\Delta_\N) \geq \excess(\A \otimes \Delta_n,\Delta_\N \otimes \Delta_n).
    \end{equation}
    Furthermore, there is (an exact) univalent simulation $\Delta_\N \otimes \Delta_n \to \Delta_\N$ given on states by the linear map
    \begin{equation}
        \begin{aligned}
            \ell^1 \otimes \R^n &\to \ell^1 \\
            (\omega_i)_{i \in \N} \otimes \delta_j &\mapsto \bigl( \omega_{i/j} \delta_{j | i} \bigr)_{i \in \N},
        \end{aligned}
    \end{equation}
    where $(\omega_i)_{i \in \N}$ is an element of $\ell^1$, $\delta_j$ is a basis vector in $\R^n$, and $\delta_{j | i}$ returns $1$ if $j$ divides $i$ and $0$ otherwise.
    Using basically the same linear map for effects, merely extended to $\ell^\infty \otimes \R^n \to \ell^\infty$, gives the required simulation.
    Consequently, we have
    \begin{equation}
        \excess(\Delta_\N \otimes \Delta_n , \Delta_\N) = 0,
    \end{equation}
    so that, by Inequality \eqref{eq:emb_extension} and \cref{prop:excess_triangle}, we get
    \begin{equation}
        \begin{aligned}
            \excess(\A,\Delta_\N) &\geq \excess(\A \otimes \Delta_n,\Delta_\N \otimes \Delta_n) + \excess(\Delta_\N \otimes \Delta_n , \Delta_\N) \\
            &\geq \excess(\A \otimes \Delta_n,\Delta_\N)
        \end{aligned}
    \end{equation}
    as required.
\end{proof}

\begin{theorem}
    \label{thm:inf_excess}
    The function $\excess(\ph, \Delta_\N) : \G \to \R$, which maps a GPT system to its excess within $\Delta_\N$, is order-preserving.
    That is, we have
    \begin{equation}
        \A \preceq \B  \quad \implies \quad  \excess(\A,\Delta_\N) \leq \excess(\B,\Delta_\N) .
    \end{equation}
\end{theorem}

\begin{proof}
    By \cref{def:cont_preorder}, \cref{cor:excess_monotone} (specialized to $\C = \Delta_\N$), and \cref{lem:clas_no_excess} respectively, we have
    \begin{equation}
        \begin{aligned}
            \A \preceq \B  \quad &\iff \quad  \exists \, n  \; \text{ such that } \; \A \hookrightarrow \B \otimes \Delta_n , \\
            &\;\implies \quad \exists \, n  \; \text{ such that } \; \excess(\A ,\Delta_\N) \leq \excess(\B\otimes \Delta_n,\Delta_\N) \\
            &\iff \quad  \excess(\A,\Delta_\N) \leq \excess(\B,\Delta_\N).
        \end{aligned}
    \end{equation}
    This gives us the statement of the theorem.
\end{proof}

\begin{remark}[Excess within a finite simplex]
    An important feature of the monotone $\excess(\A, \Delta_{\mathbb{N}})$ is that it can sometimes be computed by considering the $\A$-excess within $\Delta_{k}$ for a finite $k$.
    This happens whenever its value $\excess(\A, \Delta_m)$ remains constant for all classical systems with $m$ larger than $k$.
    For instance, the noncontextual GPT system corresponding to the toy bit in Spekkens toy theory \cite{Spekkens2007} is embedabble within $\Delta_4$ with vanishing error, and hence for any larger $m$ its excess within $\Delta_m$ has to remain $0$.
\end{remark}

\subsection{Parity oblivious multiplexing success probability with free classical resources as a measure of contextuality}\label{sec:POM}

Having introduced the hierarchy of contextuality as well as a new contextuality measure in the form of the classical excess, we now build a measure of generalized contextuality based on the parity oblivious multiplexing (POM)~\cite{Spekkens2009} game. The latter is  a protocol that is powered by preparation contextuality and has raised significant attention in recent years \cite{Banik2015,Chailloux2016,Ghorai2018,Saha2019,Ambainis2019,Tavakoli2021,Catani2024,khoshbin2023}.

In the \newterm{$\bm{n}$-bit Parity Oblivious Multiplexing (POM) game}, Alice is given an $n$-bit string, whose possible values she encodes in a GPT system $\A$.
In particular she chooses a state $\omega_{ x}$ for possible string value $ x$.
The system is then transmitted to Bob is additionally given an integer $y \in \{1,2,\ldots, n\}$. 
Bob's task is to guess the value $x_y$ of the y$^{\rm th}$ bit of $ x$ by performing a two-outcome measurement $M_y$. 
His guess is denoted by $b \in \{0,1\}$.  

The task has a supplementary constraint called parity obliviousness (PO).
Namely, Alice cannot communicate the parity of the string, denoted by $p( x)$, to Bob. 
The PO condition can be phrased in terms of Alice's chosen states as 
\begin{equation}\label{eq:PO}
    \sum_{ x| p( x) = 0} \omega_{ x} =  \sum_{ x| p( x) = 1} \omega_{ x}.
\end{equation}
Note that the string ${x}$ is assumed to be uniformly distributed, so that the left-hand side of \cref{eq:PO} can be interpreted (up to a factor of $2^{n-1}$) as the average state sent by Alice given that the characters of ${x}$ add up to $0$ modulo $2$ and similarly for the right-hand side.

The average success probability is given by
\begin{equation}
    \sum_{ x , y} \frac{1}{2^{n} n} p(b =  x_y \,|\, \omega_{ x}, M_y),
\end{equation}
where $2^{-n}$ is the probability of Alice receiving string $ x$, $1/n$ is the probability of Bob receiving $y$, and ${p(b =  x_y \,|\, \omega_{ x}, M_y)}$ is the probability that Bob's guess is correct in a given run.

It is shown in~\cite[Theorem 2]{Spekkens2009} that if a system $\A$ is noncontextual, then the \newterm{optimal} (i.e.\ maximal) \newterm{success probability} for the $n$-bit POM, denoted by $\mathfrak{p}_n$, is upper bounded by 
\begin{equation}
    \mathfrak{p}_n(\A) \leq \frac{n+1}{2n} =  \mathfrak{p}_n(\Delta_\infty). 
\end{equation}

The success probability for POM can be thus used as a witness of contextuality.

We now prove that for the embeddability preorder the POM success probability is also order-preserving.

\begin{lemma}\label{lem:embed_POM_prob}
    Let $n$ be any integer.
    If $\A$ is embeddable within $\B$, then we have $\mathfrak{p}_n(\A) \leq \mathfrak{p}_n(\B)$.
\end{lemma}

\begin{proof}
    Let us consider  any choice of states $\omega_x$ (for $x \in \{0,1\}^n$) which obey the parity obliviousness constraint \eqref{eq:PO}.
    Since we have $\A \hookrightarrow \B$, for any choice of effects $e_y$ and states $\omega_x$ on $\A$ obeying the parity obliviousness constraint, the states $\Gamma(\omega_x)$ on $\B$ also obey it by the linearity of the single-valued state simulation $\Gamma$.
    Moreover, the empirical adequacy (\cref{eq:sim_emp_adeq}) of the exact simulation implies $e_y \cdot \omega_x = \Gamma(e_y) \cdot \Gamma(\omega_x)$. 
    Thus, any strategy on $\A$ can be implemented on $\B$ with equal success probability. 
    This implies $\mathfrak{p}_n(\A) \leq \mathfrak{p}_n(\B)$.
\end{proof}

However, it is not a resource monotone in our resource theory of GPT-contextuality.

\begin{lemma}
    The optimal success probability for POM is not a monotone with respect to the hierarchy of contextuality $\preceq$.
\end{lemma}
\begin{proof}
    We define  a \textit{noisy bit} to be any system with the same states as the classical bit $\Delta_2$, but whose effects form a proper subset of $\Delta_2^*$.
    In particular, a noisy bit has no measurement that distinguishes the two extremal states $\delta_1$ and $\delta_2$ with certainty. 
    Noisy bits can be parametrized by a noise parameter $\alpha \in (0, 1/2]$, giving rise to a GPT system $\Delta_2^{\alpha}$ whose extremal effects $\chi_1$ and $\chi_2$ satisfy
    \begin{equation}
        \chi_i \cdot \delta_j = 
            \begin{cases}
                1 - \alpha  & \text{if } i=j \\
                \alpha  & \text{if } i \neq j.
            \end{cases}
    \end{equation}
    By the desiderata~\ref{it:classical_free} and \ref{it:sim_free}, a noisy bit $\Delta_2^{\alpha}$ and a bit $\Delta_2$ must satisfy $\Delta_2^{\alpha} \preceq \Delta_2$ and $\Delta_2 \preceq \Delta_2^{\alpha}$ since they are both noncontextual. The optimal strategy for POM given a (noisy) bit consists of encoding the first bit $x_1$ of the string $x$ in its extremal states $\delta_1$ and $\delta_2$ \cite[Lemma 1]{Spekkens2009}. 
    With probability $1/n$ Bob is asked to guess this specific bit, which he can do successfully with probability $(1- \alpha)$. 
    With probability $(n-1)/{n}$ Bob is asked to guess any of the remaining bits from $x$, for which he can do no better than a uniformly random guess, i.e.\ the success probability is $1/2$. 
    Hence, the optimal success probability is 
    \begin{equation}
        \mathfrak{p}_n(\Delta_2^{\alpha}) = \frac{1-\alpha}{n} + \frac{n-1}{2n} = \frac{ n + 1 - 2\alpha}{2n},
    \end{equation}
    which is strictly smaller than $\mathfrak{p}_n(\Delta_2) = (n+1)/2n$ whenever the noise parameter $\alpha$ is greater than $0$.
    This means that $\Delta_2 \preceq \Delta_2^{\alpha}$ does not imply $\mathfrak{p}_n(\Delta_2^{\alpha}) \geq \mathfrak{p}_n(\Delta_2)$. 
\end{proof}

We can see that the POM success probability cannot serve as a monotone because the POM game does not allow the use of classical systems for free.
This can be remedied by the so-called generalized yield construction \cite[Section 3.1]{gonda2023monotones}.

\begin{definition}
    \label{def:modified_pom}
    Given a GPT system $\A$ we define the optimal \newterm{POM success probability with free classical resources} as:
    \begin{align}
        \pomyield(\A) \coloneqq \sup_{d \in \Nl} \mathfrak{p}_n (\A \otimes \Delta_d).
    \end{align}
\end{definition}

\begin{proposition}
    $\pomyield$ is a monotone on the resource theory of GPT-contextuality.
\end{proposition}

\begin{proof}
    Let us assume $\A \preceq \B$, which implies that there is an $m \in \Nl$ such that there is a univalent simulation $\A \otimes \Delta_{d} \to \B \otimes \Delta_{d \cdot m}$ for every $d \in \Nl$. 
    By Lemma~\ref{lem:embed_POM_prob}, we obtain $\mathfrak{p}_n(\A \otimes \Delta_d) \leq \mathfrak{p}_n(\B \otimes \Delta_{d \cdot m})$ for every $d \in \Nl$, which yields the  inequality below: 
    \begin{equation}
        \sup_{d \in \Nl}\mathfrak{p}_n(\A \otimes \Delta_d) \leq \sup_{d \in \Nl}\mathfrak{p}_n(\B \otimes \Delta_{d \cdot m}).
    \end{equation}
Moreover,
\begin{align}
      \sup_{d \in \Nl}\mathfrak{p}_n(\B \otimes \Delta_{d \cdot m})= \sup_{d \in \Nl}\mathfrak{p}_n(\B \otimes \Delta_{d}) , 
\end{align}
    which implies
    \begin{equation}
         \pomyield(\A) \leq \pomyield(\B) ,
    \end{equation}
     using \cref{def:modified_pom}.
\end{proof}

This proof is an instance of the proof of the generalized yield construction in \cite[Theorem 4.21]{gonda2021resource}, so that $\pomyield$ can be meaningfully interpreted as a (generalized) yield monotone.

\section{Discussion}\markboth{DISCUSSION}{}
\label{sec:discussion}

This section is divided into two subsections. In the first, we discuss whether our resource-theoretic perspective on contextuality could offer an avenue to understand contextuality as an expression of information erasure.
In the second, we discuss how our work relates to the main studies on which it builds, as well as to the resource theory of generalized contextuality developed in \cite{Duarte2018}.

\subsection{Contextuality and information erasure}
\label{sec:contextualityerasure}

The resource-theoretic perspective on contextuality \cite{Spekkens2009,Hameedi2017,Schmid2018,Saha2019,LostaglioSenno2020,lostaglio2020certifying,Yadavalli2020,Flatt2021,Roch2021,Duarte2018,plavala2022incompatibility} highlights the following aspect: the existence of a noncontextual model (i.e.\ noncontextuality) of a physical system imposes constraints on the type of information that this system can carry.
This is in essence why one can show that its violation (i.e.\ contextuality) can act as a resource for information-theoretic tasks.
Our resource theory of GPT-contextuality aims to give a more precise description of the nature and properties of this information-theoretic manifestation of contextuality as compared to previous works that often focus merely on the distinction between contextual and noncontextual systems.

From a more foundational perspective, a contextual ontological model gives a classical description of an operational theory that postulates ontological distinctions that are in principle inaccessible. 
There are two types of responses to the phenomenon of contextuality (e.g.\ that of quantum systems).
One is to reject ontological models as the appropriate framework to explain the operational features of the theory.
The other is to accept the possibility that a contextual ontological model provides a feasible realist explanation of the empirical predictions of the theory. 

The latter, however, begs for an explanation of the fine-tuning associated to it \cite{CataniLeifer2020}.
How can we understand that ontological distinctions disappear at the operational level?
One could explain the presence of fine-tuning implied by contextuality if there was a physical process responsible for this disappearance.
It would map the ontological theory to the operational theory, effectively erasing the extra information that is supposed to be inaccessible at the level of the operational theory.
We refer to this hypothetical process, associated to any ontological model, as the \emph{revelation map}, for short.

One could then ask:
can revelation maps explain the detailed information-theoretic properties of contextuality?
To tackle this question, we may take our resource theory of GPT-contextuality as an expression of the information-theoretic properties of contextuality.
In this sense, the question then becomes:
can the resource theory of GPT-contextuality be alternatively interpreted as a resource theory of erasure for revelation maps?
We cannot currently answer this question, since a resource theory of erasure for revelation maps has not yet been developed. 
Nevertheless, we believe the question is worthwhile and we outline a few relevant considerations below.

\paragraph{Explaining fine-tunings as emergent from yet undiscovered physical mechanisms.} 
The presence of fine-tunings in a contextual ontological model gives a conspiratorial connotation to the realist explanation of the theory it provides.
One attempt to justify the presence of such fine-tunings, by searching for new physical mechanisms, dates back to Valentini's variant of Bohmian mechanics \cite{Valentini1991}. 
There, he introduces a notion of quantum equilibrium as the reason why superluminal signaling does not manifest in quantum theory, despite the nonlocality of its underlying ontological model. 
This picture predicts that outside of the quantum equilibrium, it is possible to observe faster than light signaling.
    Therefore, the fine-tuned nature of no-signaling in Bohmian mechanics is explained just as an emergent feature of the quantum equilibrium and it is not universally valid. 
    We cannot avoid noticing how radical such explanations of fine-tunings rooted in undiscovered physical mechanisms are. 
    They imply that an established physical principle, such as the principle of no-signaling, is violated at the fundamental level.
    In the case of contextuality, 
    the physical mechanism explaining the emergence of the operational equivalences would entail the existence of measurements that can distinguish behaviours that are deemed indistinguishable by quantum theory.

\paragraph{Explaining contextuality through information erasure.}
    
    In Valentini's work, the quantum equilibration process is responsible for the emergence of no-signalling\,---\,the fine-tuned feature associated with nonlocality.
    What hypothetical physical mechanism could be responsible for the emergence of operational equivalences\,---\,the fine-tuned feature associated with contextuality? 
    
    It would have to be a process that involves a kind of information erasure.
    The information erased is the information about distinctions at the ontological (e.g.\ fundamental) level, which cannot be stored in systems of the operational (e.g.\ effective) theory that lacks these distinctions. 
    By Landauer's principle, we can then associate an increase in entropy between the fundamental and effective levels.
    Such a process of information erasure would not only provide an explanation for the problematic fine-tuning associated with contextuality but would also be associated to a potentially detectable heat dissipation. 
    This heat would signify that, indeed, there are distinctions at the fundamental level which are not present at the effective level.
   One could even hypothesise that the information erasure is a physical process occurring over time.
   That is, during the preparation of a quantum system there may be a timescale before which the system is described by the fundamental (and noncontextual) theory.
   At longer timescales, once the erasure has occurred, the system can only be described by the effective (contextual) theory.

\paragraph{Empirical nature of the fundamental theory.}

In the above hypothetical account of Bohmian mechanics, the theory of Bohmian particles is seen as a more fundamental one than quantum theory. 
    
    In order for a similar explanation of the fine-tuning associated with contextuality to be testable, the more fundamental theory must be in principle accessible.
    Only then would we expect entropy increase and heat transfer as a result of the erasure in the revelation map.
    This is in accordance with the view presented by Müller and Garner in \cite{MullerGarner2021}, but it is not in opposition to the standard use of ontological models to study contextuality \cite{Spekkens2005}.
    
    Indeed, if one tries to find evidence for the physical erasure process we hypothesize here and fails, this gives more credence to the common assumption that any ontological distinctions, which are not present in the operational theory, are in principle indistinguishable.
    On the other hand, if such evidence is found, the detailed study of contextuality in quantum theory provides invaluable clues towards finding the more fundamental theory, from which quantum theory emerges.

\paragraph{Ontological models as physical erasure processes.}
    In case the fundamental theory in question is classical, one is justified in using standard ontological models to describe how the effective theory is simulated by the fundamental one.
    Namely, an ontological model of (the operational theory associated to) a GPT system $\A$ (\cref{sec:op_th_for_gpt}) is equivalent to a simulation of $\A$ by a classical system $\Delta_n$. 
    The classical system $\Delta_n$ is interpreted as the fundamental system and $\A$ as the effective description thereof.
    A simulation $\A \to \Delta_n$ specifies all behaviours of the fundamental theory $\Delta_n$ that are compatible with a given behaviour of the effective theory $\A$.

    Whether we can interpret $\A$ as a description of $\Delta_n$ that emerges as a result of a revelation map from $\Delta_n$ to $\A$ depends on the kind of revelation maps that we consider.
    For example, if it supposed to be a pair of partial functions (one for states, one for effects) $f : \Delta_n \to \A$, then the associated simulation would be given by preimages under $f$.
   We could obtain an arbitrary simulation in this way.
    
    However, many simulations cannot be obtained like this if we further require the revelation map $f$ to be a \emph{physical GPT transformation} (\cref{def:GPT_transf}).
    Namely, we say that a simulation $(\Gamma, \Theta)$ of $\A$ by $\Delta_n$ admits of a \emph{physical realisation} if there exists a linear map $M : \Delta_n \to \A$, such that $\Gamma(\omega) \subseteq M^{-1}(\omega)$ for all $\omega \in \Omega_\A$ and $M^*(e) \in \Theta(e)$ for all $e \in E_\A$. See~\cref{app:phys_sim} for more details.
    
    For instance, while there is a noncontextual ontological model of the Spekkens toy bit~\cite{Spekkens2007} by two classical bits, i.e.\ a univalent simulation by $\Delta_4$, there cannot exist a physical realisation of this simulation, as we prove in~\Cref{app:phys_sim}.
    
    Nevertheless, there are ontological models with physical realisations.
    For any GPT system $\A$ with a finite set $\{\omega_i\}_{i=1}^n$ of convexly extremal states, we can define the Holevo-Beltrametti-Bugajski (HBB) model of $\A$ (\cref{ex:HBB}), which may be contextual. 
    In \cref{lem:holevo_map}, we show that every HBB model, viewed as a simulation of $\A$ by a classical GPT system $\Delta_n$, admits of a physical realisation $M : \Delta_n \to \A$.
    At the level of states, $M$ maps a given extremal state of the simplex $\Delta_n$ to the corresponding extremal state $\omega_i$ of $\Omega_\A$.

\paragraph{Erasure associated with ontological models.} 
    By analyzing physical GPT transformations like $M$ above, we can understand why information erasure in the context of revelation maps cannot be the same as the erasure in the context of thermodynamics. 
    While the latter can be understood as a coarse-graining operation, physical maps associated to contextual ontological models are not necessarily of this kind.
    
    To see this, consider a contextual ontological model of a GPT system $\A$, which is given by preimages under a physical GPT transformation $f : \Delta_d \to \A$. 
    Since $\A$ is contextual, there is no coarse-graining (i.e.\ idempotent) process on $\Delta_d$ whose image is $\A$ \cite[Lemma 11]{MullerGarner2021}. 
    Hence, the physical GPT transformation $f$ cannot be understood simply as a coarse-graining operation, such as the erasure that occurs between fundamental and effective levels in statistical mechanics.
    
    A more concrete perspective is offered by the HBB map $ M $ which explicitly models the preparation of the system $\A$ via the classical control $\Delta_n$. 
    The classical degree of freedom is used to select the extremal state of $\A$ to be prepared. 
    Thus, while $M$ can be used to model a physical process in the laboratory (e.g.\ the preparation of $\A$ conditioned on a classical variable), it does not correspond to a simple coarse-graining process applied to the classical degree of freedom.

\subsection{Relation with previous works on contextuality and GPTs}
\label{sec:previousworks}

In the following section we relate the resource theory of generalized contextuality introduced in this work to previous works on generalized contextuality in the frameworks of GPTs and contrast it to the existing proposed resource theory of contextuality of~\cite{Duarte2018}.

\paragraph{On the framework of M\"{u}ller and Garner.}
    On the technical level, our work makes use of the tools developed in \cite{MullerGarner2021}.
    For instance, the concept of univalent simulations of GPT systems plays an especially important role. 
    Even though we use similar tools, our goal is quite different. 
    M\"{u}ller and Garner use the concept of univalent simulations to formulate a notion of generalized noncontextuality that applies to generic effective and corresponding fundamental theories. 
    The effective theory is the theory emerging from the operational statistics of the experiment and the fundamental theory is a fine-grained theory of the effective ones (so, not just the simplicial theory, as in the case of standard noncontextuality).

    They use such noncontextuality as a plausibility criterion for testing the validity of the fundamental theory given the effective theory. 
    This approach becomes particularly meaningful when the fundamental theory is assumed to be quantum theory, since it allows the authors to provide experimental tests of quantum theory.
    Related works which employ an alternative definition of noncontextuality  are the ones of Gitton and Woods \cite{GittonWoods2020} and \cite{GittonWoods2022}.
    For ongoing debate on the merits of the different approaches see \cite{GittonWoods2022} and \cite{schmid_addressing_2024}.

\paragraph{On simplex embeddability.} 
    Other influential sources of inspiration for our work are \cite{Schmid2019,selby2023accessible,Selby2024}. As we already mentioned in the introduction, \cite{Schmid2019} introduces the notion of simplex embeddability as the geometrical criterion to assess whether a GPT system is noncontextual.\footnotemark{}\footnotetext{As a note our notion of the canonical ontological model associated to a GPT system differs from the notion of ontological model of a GPT system of~\cite{Schmid2019} since we define the preparations to be the set of ensembles of states, whereas they define the preparations to be the set of states (hence every ontological model of a GPT system is noncontextual in their work).}\footnotemark{} 
    \footnotetext{Another work that is closely related to \cite{Schmid2019} is that of Shahandeh \cite{shahandeh2021contextuality}. 
    However, the criterion of classicality introduced there additionally requires the simplex to have the same dimension of the GPT system. 
    We do not adopt this approach.
    For example, it deems the rebit stabilizer theory to be nonclassical. 
    However, the rebit stabilizer theory, in the prepare-and-measure scenario, admits of a noncontextual ontological model given by the Spekkens toy theory \cite{Spekkens2007}, as has been shown in several \mbox{works \cite{Spekkens2016,CataniBrowne2017,CataniBrowne2018}}.}%
    The criterion is extended to accessible GPT fragments in \cite{selby2023accessible}, where the latter correspond to more general mathematical objects than GPTs and characterize generic prepare-and-measure experimental setups. 
    The work of \cite{Selby2024} provides an algorithm for testing contextuality in any prepare-and-measure scenario. In particular, if it exists, it returns an explicit noncontextual model for the scenario and, if not, it provides the minimum amount of (depolarizing) noise which would be required until a noncontextual model would become possible. They call this measure of contextuality the \textit{robustness of nonclassicality}. The latter is related to our notion of error of univalent simulation, but it is different insofar as it requires a specific noise channel (e.g.\ depolarizing noise channel), whereas the error of univalent simulation does not require that. Moreover, the robustness of nonclassicality is not defined as a measure within a well defined resource theory. 

\paragraph{On accessible GPT fragments.}

It has been argued in \cite{selby2023accessible} that in realistic experiments GPT systems are generally not appropriate to study contextuality.
       In particular, the following discrepancy arises:
       A GPT system obtained from an operational description of the experiment provides states and effects observable \emph{in practice} in this experiment, but does not limit the \emph{in principle} implementable ones.
       However, the preparation equivalences preserved by a preparation-noncontextual ontological model (see \cref{sec:prep_cont}) are \emph{in principle} operational equivalences.
        
       When studying preparation contextuality only, this discrepancy can be dealt with.
       Namely, one can first set-up the effect space $E_\A$ in a non-operational way, as a hypothesized set of in principle possible effects.
       Importantly, one has to identify each actual measurement outcome arising in the experiment with a (distribution over) the elements of $E_\A$.
       Their means specify a subset $\tilde{E}_\A \subseteq E_\A$ of the so-called \emph{accessible} effects, which however \emph{play no role} in the study of preparation contextuality. 
       The state space $\Omega_\A$ is the set of equivalence classes of preparation procedures, but only with respect to in principle operational equivalence.
       That is, states that are indistinguishable by accessible effects, but could be distinguished by a hypothesized effect in $E_\A \setminus \tilde{E}_\A$ ought not be identified in $\Omega_\A$.
       With this set-up, one can study preparation contextuality as we discuss shortly in \cref{sec:GPT_simulations} in terms of the GPT system $(\Omega_\A, E_\A)$.

       Such a treatment is not possible when studying both preparation and measurement contextuality simultaneously.
       One has to replace the notion of a GPT system by the notion of an \emph{accessible fragment} of a GPT system introduced in \cite{selby2023accessible}.
       Simply put, this is a pair of a GPT system $(\Omega_\A, E_\A)$ together with a subsystem $(\tilde{\Omega}_\A, \tilde{E}_\A)$ thereof, which importantly can violate the non-degeneracy condition \eqref{eq:non-degenerate}, and therefore are not necessarily GPT system.
       The enveloping system describes states and effects that can arise in principle, while the subsystem describe those that are achievable in the particular experiment, whose contextuality we are interested in.
       
       While it would be arguably more appropriate to develop our hierarchy of contextuality for accessible fragments instead of plain GPT systems, we chose against this option for two reasons.
       Firstly, this would make the approach superficially more complicated, even though all the key ideas can be presented in the simpler framework of GPT systems.
       Secondly, one can easily accommodate our approach to the case of accessible fragments.
       The only difference is that the preparation (and outcome) equivalences have to be adjusted.
       For example, two preparations (i.e.\ distributions over $\tilde{\Omega}_\A$) would be deemed equivalent if their means cannot be distinguished by any effect in $E_\A$. 

\paragraph{On the resource theory of contextuality due to Duarte and Amaral.} %
Another related work is~\cite{Duarte2018}, where the authors also provide a resource theory of generalized contextuality in prepare-and-measure scenarios using a `black box' framework. In~\cite{Duarte2018} the basic object consists of a \emph{prepare-and-measure scenario} $\mS$ together with a \emph{behaviour} $B$. A scenario is a set $\mS = \{\P,\M,\mD,\E_P,\E_M\}$ where $\P$ is a set of preparations, $\M$ a set of measurements, $\mD$ a set of outcomes, $\E_P$ the operational equivalences for preparations and $\E_M$ the operational equivalences for measurements. 

A behaviour $B$ is a set of conditional probabilities:
\begin{align}
    B = \{p(k|j,i)\}_{i \in \mI, j \in \mJ, k \in \mK}
\end{align}
corresponding to the probability of outcome $d_k$ given preparation $P_i$ and measurement $M_j$. $\mI, \mJ$ and $\mK$ are sets of labels for the preparations, measurements and outcomes, respectively.

We note that the sets of operational equivalences $\E_P$ and $\E_M$ of a given object are defined independently of the probability distributions. As such an object (scenario + behaviour) may contain operational equivalences which are not implied by the behaviour. Alternatively, there may be operational equivalences present in the behaviour which are not included in the scenario.
In the cases where all ensembles of procedures are valid procedures in the scenario, and where the operational equivalences  implied by behaviour are those of the scenario, then the object defined is equivalent to a GPT system.

The free objects of~\cite{Duarte2018} are the noncontextual behaviours, namely those which admit of a generalized noncontextual model.

The operations of \cite{Duarte2018} are given by pre- and post-processing on the preparations, measurements and outcomes and correspond to channel simulations. The free operations correspond to channel simulations such that all the operational equivalences of the simulated system are images of operational equivalences of the simulating system through the simulation map.

Let us consider an example where both resource theories can be applied to highlight the differences between the resource theory of~\cite{Duarte2018} and the one presented in this work.

We consider a scenario $\mS = \{\P,\M,\mD,\E_P,\E_M\}$ and a scenario $\tilde \mS = \{ \tilde \P,\tilde \M, \tilde \mD,\tilde \E_P,\tilde \E_M\}$ where,
\begin{alignat}{5}
    &\P = \{0,1\} , \quad &&\M = \{0\}, \quad  &&\mD = \{0,1\} , \ &&\E_P = \emptyset , \ &\E_M = \emptyset , \\
     &\P = \{0,1,2\} , \quad &&\M = \{0\}, \quad &&\mD = \{0,1,2\} , \quad &&\E_P = \emptyset , \quad &\E_M = \emptyset .
\end{alignat}
The behaviours on $\mS$ and $\tilde \mS$ are $B$ and $\tilde B$ respectively:
\begin{align}\label{eq:bit_trit}
    p(k|j = 1,i) &= \delta_{i,k} , \ i, k \in \{0,1\}, \\
    p(\tilde k|\tilde j = 1, \tilde i) &= \delta_{\tilde i, \tilde k}, \ \tilde i, \tilde k \in \{0,1,2\}.
\end{align}
These are the behaviours obtained by preparing a bit (resp. trit) in one of its extremal states and then implementing the canonical two outcome (resp. three outcome) measurement which returns one of the outcomes with certainty. 

In the framework of GPTs $(\mS, B)$ corresponds to a bit and $(\tilde \mS, \tilde B)$ to a trit (assuming one also allows all ensembles of preparations and measurements). 

In the resource theory of GPT contextuality presented in this work we have that the bit and the trit are equivalent as resources: $\Delta_2 \succeq \Delta_3$ and $\Delta_3 \succeq \Delta_2$.

However in the resource theory of~\cite{Duarte2018} this is not the case. We now show explicitly that there is no free operation $T$ taking $(\mS,B)$ to $(\tilde \mS, \tilde B)$. Since there is only one measurement in both scenarios we drop the indices $j$ and $\tilde j$. Let us assume the existence of stochastic maps $q_O(\tilde k|k)$ and $q_P(i|\tilde i)$ which map $B \to \tilde B$:
\begin{align}
    p(\tilde k|\tilde i) = \sum_{i,k = 0}^1 q_O(\tilde k|k) p(k|i) q_P(i|\tilde i) 
\end{align}
Now using the definitions of $B$ and $\tilde B$ of~\cref{eq:bit_trit} we obtain:
\begin{align}
    \delta_{\tilde i, \tilde k} = \sum_{i = 0}^1 q_O(\tilde k|i) q_P(i|\tilde i)  . 
\end{align}
Let us denote by $Q_O$ the $3 \times 2$ matrix with entries $[Q_O]_{\tilde k, i} = q_O(\tilde k|i) $ and $Q_P$ the $2 \times 3$ matrix with entries $[Q_P]_{i, \tilde i} = q_P(i|\tilde i)$. Then the previous equation is equivalent to the matrix equation:
\begin{align}
    \I_3 = Q_O \cdot Q_P
\end{align}
This yields a contradiction since the decomposition of a $3 \times 3$ matrix $A$ into matrix product  $A = B \cdot C$ with $B$ a $3 \times 2$ matrix and $C$  a $2 \times 3$ matrix implies  $C$ has rank at most 2. However $\I_3$ is a rank 3 matrix. This shows that there is no free operation  $(\mS, B) \mapsto (\tilde \mS, \tilde B)$ in the resource theory of~\cite{Duarte2018}.
Notice that the reason why in our resource theory the bit and the trit{\,---\,}and any other noncontextual theory{\,---\,}are equivalent is that our free operations include the access to the classical system. This is the origin of the difference with the resource theory of~\cite{Duarte2018}.

An additional contribution of the present work is the following. In~\cite{Duarte2018} the authors define a notion of composite of two behaviours which is equivalent to the minimal composite of~\cref{def:minimal_tensor} in our framework. They then show that this composite is consistent with their resource theory, in the sense that the free objects are closed under this composite. In the present work we show that any consistent composite of GPT systems, not just the minimal composite of~\cref{def:minimal_tensor} is compatible with the resource theory of contextuality.

\section{Conclusion} \markboth{CONCLUSION}{}
\label{sec:conclusion} 

Generalized contextuality, a leading notion of nonclassicality, is of crucial importance both in the foundations of quantum theory and in quantum information processing. Despite this, there is no complete characterisation of generalized contextuality as a resource. In this work we address this shortcoming.

Based on recent developments formulating noncontextuality of GPT systems via simplex embeddability, we have defined a resource theory of contextuality of GPT systems in prepare-and-measure scenarios. The free resources are the noncontextual systems and the free operations are univalent simulations with free access to classical systems. Using these notions, we motivate a hierarchy of contextuality for GPT systems. 
A new contextuality monotone arises naturally from our considerations{\,---\,}the classical excess, which expresses the minimum error of a univalent simulation by the countably-infinite classical system.  
We have also shown how a standard witness of contextuality, in the form of the POM success probability, can be used to define a contextuality monotone. 

A key feature of the resource theory presented in this work is that the free operations include access to classical systems. Since every noncontextual system can be univalently simulated by a classical system this entails that all noncontextual systems are equivalent in the hierarchy. 
Without access to free classical systems the preorder which emerges is the embeddability preorder, for which generic noncontextual systems are inequivalent. 
Moreover we show, by assuming that the only way of composing GPT systems with classical ones is the standard one, that  any consistent composition of GPT systems is consistent with the hierarchy: if $\A_1 \preceq \B_1$  and $ \A_2 \preceq  \B_2$ then $\A_1 \tildetensor \A_2 \preceq  \B_1 \tildetensor \B_2$.

We have also discussed how GPT simulations could describe a physical process of information erasure that would explain the fine-tuning associated with contextuality\footnotemark{} in a similar way to how the quantum equilibration process proposed by Valentini explains the fine-tuning associated with nonlocality in Bohmian mechanics. 
\footnotetext{A recent work that connects preparation contextuality and information erasure is \cite{Montina2023}, where the authors show that any ontological model reproducing the statistics of a sequential protocol involving incompatible projective measurements involves more information erasure than what operational quantum theory predicts. 
This fact is strictly related to the presence of preparation contextuality, as the same final quantum state in the protocol is represented by two different ontic distributions.  
This result can be seen as an example of fine-tuning of information erasure (i.e.\ more erasure at the ontological level than at the operational level), which, despite being interesting, differs from our idea of viewing contextuality as arising from a process of information erasure that explains the operational equivalences of distinct ontological representations.}%
We have argued that the information erasure in this case would be different than a simple coarse-graining. 
An interesting avenue for future research would be to further characterize this kind of information erasure and propose a test to detect heat dissipation in experiments manifesting contextuality. 
Even though such a proposal is undoubtedly radical, we believe that it could  explain the fine-tuning associated with contextuality without abandoning the ontological models framework.

Another direction for further investigation is to develop a solid interpretation and establish the use-cases of the resource theory of GPT-contextuality that we provide. 
It is known that a resource theory of a given notion need not be unique. We have examples of this fact, like the resource theories of entanglement based on LOCC~\cite{Bennett1996} and LOSR~\cite{Buscemi2012,SchmidLOSR2023} free operations, respectively. 
This does not mean that only one of them is \qmarks{the correct resource theory of entanglement}.
Rather, these resource theories may be applicable in different contexts or for different purposes. 
In the case of entanglement, one can say that LOCC operations are relevant for communication tasks in the context of quantum internet and LOSR operations are relevant for the study of entanglement in Bell scenarios. 

In this respect, we have discussed the relation of our work with other studies of generalized contextuality for GPTs in section \ref{sec:previousworks}. One that we did not mention is \cite{Tavakoli2021}. It would be interesting to find a relationship between our classical excess and the simulation cost of contextuality defined therein. 
Finally, it would be also interesting to extend the methods developed here to the realm of resource theories of more general fine-tunings \cite{CataniLeifer2020}, such as violations of time symmetry \cite{LeiferPusey} and bounded ontological distinctness \cite{Chaturvedi2020}.

\section*{Acknowledgments}

    The authors thank the participants of the PIMan workshop -- Orange (CA), March 2019 -- where the idea of this project originated. In particular, Luke Burns and Justin Dressel, who were part of the initial discussions on the project. 
    The authors further thank Rafael Wagner for insightful explanations regarding the approach of Duarte and Amaral \cite{Duarte2018}.
    LC thanks Farid Shahandeh and TGa thanks Markus M{\"u}ller for helpful discussions.
This project started when LC was supported by the Fetzer Franklin Fund of the John E. Fetzer Memorial Trust and by the Army Research Office (ARO) (Grant No. W911NF-18-1-0178). 
LC also acknowledges funding from the Einstein Research Unit ``Perspectives of a Quantum Digital Transformation'' and from the Horizon Europe project FoQaCiA, GA no.101070558.  
TGa acknowledges support from the Austrian Science Fund (FWF) via project P 33730-N. This research was partly funded by the Austrian Science Fund (FWF)
10.55776/PAT2839723. TGo acknowledges support from the Austrian Science Fund. 
This research was funded in whole or in part by the Austrian Science Fund (FWF) via the START Prize Y1261-N. 
This research was supported in part by Perimeter Institute for Theoretical Physics. 
Research at Perimeter Institute is supported by the Government of Canada through the Department of Innovation, Science, and Economic Development, and by the Province of Ontario through the Ministry of Colleges and Universities. 

For open access purposes, the authors have applied a CC BY public copyright license to any accepted manuscript version arising from this submission.

\bibliography{Bibliography}

\appendix

\section{On the physicality of simulations}\label{app:phys_sim}

In the framework of GPTs a physical map $\B \to \A$ is given by a linear map $M: V_\B \to V_\A$ satisfying $M(\Omega_\B) \subseteq \Omega_\A$. $M$ can be viewed as part of a preparation procedure, preparing a state of $\A$ by first preparing a state $\omega_\B$ of $\B$ and then applying $M$ to obtain $\omega_\A$ given by $M (\omega_\B)$. 
Composing this with a measurement $\{e_\A^i\}_i$ on $\A$ gives outcome probabilities $e_\A^i \cdot M (\omega_\B)$. 
However, one could also view this experiment as preparing the state $\omega_\B$ on system $\B$ and then measuring $\{e_\A^i \circ M\}_i = \{M^* e_\A^i\}_i$ which is a measurement on $\B$, hence we require that a physical map also obeys $M^* (e_\A) \in E_\B$ for all $e_\A \in E_\A$. 

\begin{definition}[Physical GPT transformation]\label{def:GPT_transf}
    A physical GPT transformation $M$ between two GPT systems, from ${\B = (\Omega_\B, E_\B, V_\B)}$ to $\A = (\Omega_\A,E_\A, V_\A)$, is given by a linear map $M : V_\B \to V_\A$ satisfying
    \begin{equation}\label{eq:GPT_transf}
        M(\Omega_\B) \subseteq \Omega_\A \qquad \text{and} \qquad  M^*(E_\A) \subseteq E_\B .
    \end{equation}
\end{definition}

One may wonder whether an exact simulation $(\Gamma,\Theta)$ of a system $\A$ by a system $\B$ can correspond to a physical transformation $M : \B \to \A$,\footnotemark{}
in the sense that the state simulation $\Gamma$ assigns, to each state $\omega \in \Omega_\A$, a subset of those states of $\B$ that could give rise to $\omega$ when $M$ is applied.\footnotetext{There is another possibility: namely, the simulation may correspond to a physical transformation $W : \A \to \B$, in the sense that the state simulation $\Gamma$ assigns, to each state $\omega \in \Omega_\A$, the set of those states of $\B$ that could give rise to $\omega$ when $M$ is applied.
We do not treat this case here.}

\begin{definition}[Physical realisation of a simulation]\label{def:phys_realisation}
    Consider a simulation $(\Gamma, \Theta)$ of a GPT system $\A$ by a GPT system $\B$ and a physical GPT transformation $M : \B \to \A$.
    We say that $M$ is a \newterm{physical realisation} of $(\Gamma, \Theta)$ iff we have $\Gamma(\omega) \subseteq M^{-1}(\omega)$ for all $\omega \in \Omega_\A$ and $\Theta(e) \supseteq \{M^*(e)\}$ for all $e \in E_\A$.\footnotemark{}
    \footnotetext{Here, $M^{-1}(\omega)$ refers to the preimage of $\omega$ under $M$ and $\{M^*(e)\}$ refers to the singleton set whose element is $M^*(e)$.}%
\end{definition}
For example, a classical GPT system $\Delta_n$ is embeddable within an $n$-level quantum GPT system $\Q_n$.
To describe the univalent simulation, consider a basis of the latter's Hilbert space, whose elements are denoted by $\ket{i}$ for $i$ ranging from $1$ to $n$.
We can let the state simulation map each extremal point $\delta_i$ of the classical state space to the quantum state $\ketbra{i}$.
The effect simulation is then uniquely determined, it must send each $\delta_i^*$ to the effect $\rho \mapsto \bra{i} \rho \ket i$.
There is a physical realisation of this simulation given by the completely dephasing map $\Q_n \to \Delta_n$ which acts as
\begin{equation}\label{eq:dephasing_map}
    \rho \mapsto \sum_{i= 1}^n \bra{i} \rho \ket i \delta_i.
\end{equation}

Requiring a simulation $(\Gamma, \Theta)$ of a GPT system $\A$ by a system $\B$ to admit of a physical realisation is a very strong condition. 
It has the following consequences.

\begin{proposition}\label{lem:univalent_non-phys}
    Consider a simulation $(\Gamma,\Theta)$ of a system $\A$ by a system $\B$ with a physical realisation $M: \B \to \A$.
    \begin{enumerate}
        \item $M$ is surjective on states, i.e.\ we have $M(\Omega_\B) = \Omega_\A$.
        \item If $\dim(V_A) = \dim(V_B)$ holds, then $(\Gamma,\Theta)$ is a univalent simulation and $M$ is an isomorphism.
    \end{enumerate}
\end{proposition}

\begin{proof}
    By \cref{def:phys_realisation}, we have $\Gamma(\omega) \subseteq M^{-1}(\omega)$ for all $\omega \in \Omega_\A$.
    In particular, the preimage of every such $\omega$ under $M$ must be non-empty, which means $M(\Omega_\B) \supseteq \Omega_\A$.
    Combining this with the first condition in \eqref{eq:GPT_transf} gives $M(\Omega_\B) = \Omega_\A$.

    Let us now assume that the underlying vector spaces of $\A$ and $\B$ have the same dimension given by $d \in \mathbb{N}$.
    We prove that $\Gamma$ is single-valued by contradiction.
    To this end, assume that there are two distinct $\gamma_1, \gamma_2 \in \Gamma(\omega)$ for some $\omega \in \Omega_\A$.
    By empirical adequacy of the simulation, we must have 
    \begin{equation}
        f \cdot \gamma_1 = f \cdot \gamma_2 
    \end{equation}
    for all $f \in \Theta(E_\A)$.
    In particular, this means that the span of $\Theta(E_\A)$ has dimension strictly less than $d$ and no effect in $\Theta(E_\A)$ can distinguish two states of $\Omega_\B$ that differ by a scalar multiple of $\gamma_1 - \gamma_2$.

    By the fact that effects of $\A$ can be distinguished by its states, and the fact that $\Omega_\A$ is a convex set (not containing the zero vector), there are $d$ states $\rho_i \in \Omega_\A$ whose span is the whole of $V_\A$.
    To distinguish these $d$ states, one needs at least $d$ linearly independent effects, let us denote a choice of these by $e_i \in E_\A$.
    By empirical adequacy, any choice of $d$ states $\sigma_i \in \Gamma(\rho_i)$ of $\B$ must also span $V_\B$ and be distinguishable by effects in
    \begin{equation}
        \Theta\left(\{e_i\}_{i=1}^d\right) \subseteq \Theta(E_\A).
    \end{equation}
    This is now a contradiction because $\Theta(E_\A)$, as we established earlier, has dimension strictly less than $d$.
    In conclusion, $\Gamma$ must be single-valued.
    An analogous argument can be used to show that $\Theta$ must also be single-valued.

    By the univalence of the simulation, we can thus think of $\Gamma$ and $\Theta$ as convex-linear maps $\Omega_\A \to \Omega_\B$ and $E_\A \to E_\B$ respectively. 
    Since $\spann(\Omega_\A) = V_\A$ and $0 \not\in \Omega_\A$ (by the second condition in \eqref{eq:zero_unit_effects}) it follows that we can actually identify the state simulation with a linear map $V_\A \to V_\B$ denoted by $\bar \Gamma$.
   
    Since $\Gamma$ is injective by Implication \eqref{eq:image_nonoverlap} (which implies that $\bar \Gamma$ is as well) and we have $V_\A \cong V_\B$ by assumption, $\bar \Gamma$ is an isomorphism. 
   
     By \cref{def:phys_realisation}, we then have $M \Gamma(\omega) = \omega$ for all $\omega \in \Omega_\A$. 
     By linearity, we can extend this equation to the span of $\Omega_\A$ to get $M \bar \Gamma = \id_{V_\A}$, which means that $M$ is the inverse of $\bar \Gamma$ since the latter is an isomorphism.
     Thus, $M$ itself must be an isomorphism.
\end{proof}

    In the case of unequal dimensions of $V_\A$ and $V_\B$, a simulation can indeed be realised also by a physical GPT transformation that is not an isomorphism. 
    One example is the simulation of the n-level classical system $\Delta_n$ by the quantum one $\Q_n$ mentioned above.
    The completely dephasing map from \eqref{eq:dephasing_map} gives a realisation thereof.

    One of the consequences of \cref{lem:univalent_non-phys} is that whenever there exists a simulation of a GPT system $\A$ by $\B$ and $\dim(V_\A) = \dim(V_\B)$ holds, the state spaces of $\A$ and $\B$ must be isomorphic as convex sets.
    However, the same is not true for the effect spaces, as the following example shows.

\begin{example}[Physical realisation of univalent simulation for inequivalent effect spaces]
    Consider two systems $\A$ and $\bar \A$ with identical underlying vector spaces, where we also have $\Omega_\A = \Omega_{\bar A}$ and $E_{\bar A} \subsetneq E_\A$. 
    Then the inclusion maps on states and effects provide a univalent simulation of $\bar \A$ by $\A$. 
    Moreover, it is realisable by the physical GPT transformation given by the identity ${\Id : V_\A \to V_{\bar \A}}$, since we have $\Id(\Omega_\A) = \Omega_{\bar \A}$ and $\Id^*(E_{\bar \A}) \subseteq E_\A$. 
    Even though the identity is invertible as a linear map, its inverse $\Id : V_{\bar \A} \to V_\A$ is not a physical GPT transformation since $\Id^*(E_\A) \subseteq E_{\bar \A}$ is false by assumption.
\end{example}

We can apply \cref{lem:univalent_non-phys} to the (convexified) Spekkens toy bit $\mathsf{T}_2$ (with six extremal states), which has a non-contextual ontological model, i.e.\ it can be simulated by a classical GPT system, in particular by $\Delta_4$ with four extremal states. 
Since these two GPT systems have isomorphic underlying vector spaces and non-isomorphic state spaces, it follows that there can be no physical realisation $\Delta_4 \to \mathsf{T}_2$ of the univalent simulation of $\mathsf{T}_2$ by $\Delta_4$.

Interestingly, for any GPT system $\A$ with finitely many extremal states there is an ontological model thereof, which admits of a physical realisation.

\begin{lemma}\label{lem:holevo_map}
    The HBB simulation (see~\cref{ex:HBB}) of a GPT system $\A$ with $n \in \N$ extremal states by the classical system $\Delta_{n}$ has a physical realisation.
\end{lemma}

\begin{proof}
    Given a GPT system $\A$ with a finite set $\{\omega_i\}_{i=1}^n$ of extremal states, the HBB model is given by the maps $(\Gamma, \Theta)$ as defined in~\cref{ex:HBB}. 
    We first show that the opposite  ${\Gamma^{-1}: \Gamma(\Omega_\A) \to \Omega_\A}$ of $\Gamma$ (which is a function as shown in \eqref{eq:image_nonoverlap}) uniquely extends to a linear map $M: \R^{n} \to V_\A$.

    The map $\Gamma: \Omega_\A \to \Delta_{n}$ is a multivalued function defined as:
    \begin{equation}
        \Gamma(\omega) = \Set{\sum_{i=1}^n p_i \delta_i \given p_i \in [0,1], \, \sum_{i=1}^n p_i = 1, \, \omega = \sum_{i=1}^n p_i \omega_i} .
    \end{equation}
    For an arbitrary element $\delta \in \Delta_n$, which can be uniquely decomposed as
    \begin{equation}
        \delta =  \sum_i p_i \delta_i ,
    \end{equation}
    the function $\Gamma^{-1}$ is given by
    \begin{equation}
        \Gamma^{-1} \left( \delta \right) = \sum_i p_i \omega_i 
    \end{equation}
    and is in particular a convex-linear map $\Delta_n \to \Omega_\A$.
    
    Since $\Delta_{n}$ spans $\R^{n}$ and does not contain the origin, the map $\Gamma^{-1}$ uniquely extends to a linear map $M: \R^{n} \to V_\A$ satisfying $M(\Delta_{n}) = \Omega_\A$ and $\Gamma (\omega) \subseteq M^{-1}(\omega)$ for all $\omega \in \Omega_\A$.
    
    This map also satisfies $M^*(E_\A) \subseteq \Delta_n^*$, since for any $e \in E_\A$ we have
    \begin{equation}
        \forall \delta \in \Delta_n \quad \quad M^*(e) \cdot \delta = e \cdot M (\delta) \in [0,1],
    \end{equation}
    so that $M^*(e)$ is an element of $\Delta_n^*$. 
    Therefore, $M$ is a physical GPT transformation.
    
    The univalent effect simulation map is given (as a function) by $\Theta(e) = \sum_i (e \cdot \omega_i) \delta_i^*$. 
    We now have, for any ${j \in \{1,2,\ldots,n\}}$,
    \begin{equation}
        M^*(e) \cdot \delta_j = e \cdot M(\delta_j) = e \cdot \omega_j = \sum_i (e \cdot \omega_i) (\delta_i^* \cdot \delta_j) = \Theta(e) \cdot \delta_j ,
    \end{equation}
    so that $M^*(e_\B) = \Theta(e_\B)$ holds by \eqref{eq:non-degenerate_2}.
    In conclusion, this shows that $M$ is indeed a physical realisation of the HBB simulation.
\end{proof}

\end{document}